\documentclass[Review,sageh,times]{sagej}
\usepackage{moreverb,url,subfigure}

\usepackage[colorlinks,bookmarksopen,bookmarksnumbered,citecolor=red,urlcolor=red]{hyperref}

\usepackage{graphicx,soul,framed,epsfig,mathptmx,times,
amsmath,amssymb,color,flushend,gensymb,subfigure,
fancyhdr,subfigure,lipsum,amsthm, amsfonts}

\usepackage{setspace}

\newtheorem{remark}{Remark}

\newtheorem{assumption}{Assumption}
\newtheorem{theorem}{Theorem }

\def\volumeyear{2019}
\def\volumenumber{41}
\def\issuenumber{6}
\def\startpage{1750}
\def\endpage{1760}

\allowdisplaybreaks 

\definecolor{shadecolor}{rgb}{1,0.8,0.3}

\newcommand\BibTeX{{\rmfamily B\kern-.05em \textsc{i\kern-.025em b}\kern-.08em
T\kern-.1667em\lower.7ex\hbox{E}\kern-.125emX}}

\allowdisplaybreaks

\def \journalname{\textbf{PREPRINT VERSION:} Transactions of the Institute of Measurement and Control}

\graphicspath{{Figures/}}

\begin{document}

\runninghead{Kayacan}

\title{Sliding Mode Learning Control of Uncertain Nonlinear Systems with Lyapunov Stability Analysis}

\author{Erkan Kayacan}

\affiliation{Senseable City Laboratory and Computer Science \& Artificial Intelligence Laboratory, Massachusetts Institute of Technology}

\corrauth{Erkan Kayacan, Senseable City Laboratory and Computer Science \& Artificial Intelligence Laboratory, \\
Massachusetts Institute of Technology, MA  02139, USA.}
 
\email{erkank@mit.edu}

\begin{abstract}

This paper addresses to Sliding Mode Learning Control (SMLC) of uncertain nonlinear systems with Lyapunov stability analysis. In the control scheme, a conventional control term is used to provide the system stability in compact space while a Type-2 Neuro-Fuzzy Controller (T2NFC)  learns system behavior so that the T2NFC takes the overall control of the system completely in a very short time period. The stability of the sliding mode learning algorithm was proven in literature; however, it is so restrictive for systems without the overall system stability. To address this shortcoming, a novel control structure with a novel sliding surface is proposed in this paper and the stability of the overall system is proven for nth-order uncertain nonlinear systems. To investigate the capability and effectiveness of the proposed learning and control algorithms, the simulation studies have been achieved under noisy conditions. The simulation results confirm that the developed SMLC algorithm can learn the system behavior in the absence of any mathematical model knowledge and exhibit robust control performance against external disturbances.

\end{abstract}


\maketitle

\section{Introduction}

In modern control engineering, control of uncertain nonlinear systems is one of the most crucial topics \cite{Chen2016, KAYACAN2012863, 6606388, KAYACAN2017276, KAYACAN2016265, Erkanasjc}. Robust controllers intend to ensure the best control performance in the presence of the uncertainties, and high controller gain is the common method to overcome uncertainity problem in nonlinear control theory \cite{7323849}. However, these methods bring about large control actions, and demand of very powerful actuators \cite{7762162}. Moreover, the robust control performance against uncertainties is mostly obtained at a price of sacrificing the nominal control performance of the system \cite{7574310}. Therefore, a control method is required to exhibit robust control performance in the presence of uncertainties while either maintaining or improving the nominal control performance \cite{Huang2015}. As a model-free method, fuzzy logic control is an alternative solution to the model-based control approaches in the presence of uncertainties and has been applied to different real-time systems, e.g., mobile robots \cite{Castillo201219}, spherical rolling robot \cite {erkan2013} and agricultural ground vehicles \cite{erdal2015t2fnn, Kayacan2018chapter}. Even though type-1 fuzzy logic controllers are the most well known and widely used types of fuzzy logic controllers, researchers have recently focused on type-2 fuzzy logic controllers, which can only deal with high levels of uncertainty in real-time applications \cite{Maldonado2013496, castillo2014, Shing2015, MUHURI201771, Liu2017, Solis2015, 7083756}.

As a learning controller, feedback-error learning control structure was proposed in \cite{Gomi} and firstly implemented for the control of robot manipulators that was relied on the parallel structure of a conventional controller and a neuro-fuzzy based controller. The generated control signal by the conventional controller is benefited as a learning error signal to train learning algorithms. The latter learns the system behavior online and thus eliminates the former from the control of the system \cite{830951,Ruan2007770}. Since the gradient-based methods as training algorithm are complex and computationally more expensive, sliding-mode control (SMC) theory-based learning algorithms were proposed for feedback-error learning control structures \cite{Efe2000}. They can ensure faster convergence rather than the conventional learning techniques for online adjustment in neural networks, type-1 and type-2 neuro-fuzzy structures \cite{erkan2015identt2fnn}. Moreover, the stability and robustness of NFCs based on SMC theory-based learning algorithms can also be examined \cite{4084704}. The novelty of SMC theory-based learning algorithms is that the parameters are adjusted by the designed algorithm in order to fulfill a stable equation by imposing the error instead of minimizing an error function. On the other hand, SMC is inherently robust to uncertainties \cite{YIN2017282,7855666}. 

Despite their practicality and easiness to design, model-free control methods -- such as fuzzy control and neuro-fuzzy control -- are mostly criticized by model-based control community due to the fact that there is generally no convincing analytical proof of the system stability \cite{1707754}. These criticisms are reasonable since trial-and-error methods cannot be tolerated for sophisticated and expensive robotic systems, such as unmanned air vehicles and spacecraft systems. Therefore, analytical stability analysis is an overriding requirement so that researchers tend to design model-based controllers \cite{Mayne2000789, KAYACAN20141, 7525615, LAM2018390}. In order to accomplish this limitation of model-free controllers, the overall system stability proven in this paper lightens the major disadvantage of model-free control methods. In the previous studies in literature, since the overall system stability of the sliding mode learning algorithm could not be proven, the stability analysis was accomplished, and a proportional-derivative controller was used to ensure the stability in compact space \cite{ACS982}. Recently, it has been shown that if the system is a second-order system, it is possible to guarantee the stability of the overall system by adding a robust term to the control scheme \cite{7027198}. However, this stability analysis is too restrictive.

The contributions of this paper are as follows.
\begin{enumerate}
\item The stability of the overall system is proven for an nth-order uncertain nonlinear system. The overall system stability analysis removes the most significant shortcoming of model-free control methods.
\item A novel sliding surface is proposed to train the sliding mode learning algorithm for the type-2 neuro-fuzzy control (T2NFC) algorithm.
\item The learning rate of the online sliding mode learning algorithm of T2NFC and the controller gain in the conventional control term are adaptive. Therefore, it is possible to control the system without foreknowing the upper bounds of the system states and their derivatives. 
\item The developed algorithm adjusts the proportion of the lower and upper MFs in the T2NFC that allows us to handle non-uniform uncertainties in T2FLSs.
\end{enumerate}

The paper consists of four sections: The SMLC structure consisting of the T2NFC and the conventional control term, and the sliding mode control theory-based learning algorithm are given in Section \ref{section_controlstructure}. The overall system stability is given in Section \ref{sec_overalstability}. The simulation results are presented to discuss control performance and noise analysis in Section \ref{section_simulation}. Finally, some conclusions are drawn from this study in Section \ref{section_conc}.

\section{Sliding Model Learning Control Structure}\label{section_controlstructure}

\subsection{Control Scheme}
\begin{figure}[b!]
  \centering
  \includegraphics[width=4.0in]{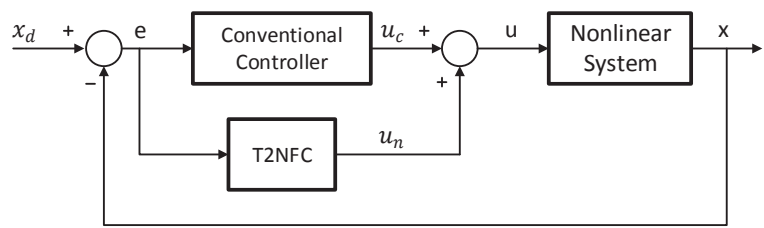}\\
  \caption{Block diagram of the control scheme}\label{fig_controlscheme}
\end{figure}

The whole control structure consists of a conventional controller and a type-2 neuro-fuzzy controller (T2NFC) as illustrated in Fig. \ref{fig_controlscheme}. The T2NFC controls the system by learning the system behavior owing to the fact its parameters are updated by a sliding mode learning algorithm. The conventional control action is added to guarantee the stability of overall system in compact space while T2NFC learns the system behavior. In this control structure, the purpose is that the output of the conventional controller converges to zero in finite time while T2NFC takes the overall control signal. The total control action applied to the system $u$ is determined as follows:
\begin{equation} \label{eq_totalcontrolaction}
u = u_{c} + u_{n}
\end{equation}
where $u_{c}$ and $u_{n}$ denote respectively the conventional control signal and the control signal generated by the T2NFC. The conventional control law is formulated as follows:
\begin{equation}\label{eq_robustterm}
u_{c}  =  k  \textrm{sgn}\left(s \right)
\end{equation}
where $s$ and $k$ denote the sliding surface and the controller gain, respectively. The controller gain $k$ is positive, i.e., $k >0$, adaptive and updated with the following rule
\begin{equation}\label{eq_robustterm_k}
\dot{k}  =  \frac{\gamma_{k} \mid s \mid } {2}
\end{equation}
where $\gamma_{k}$ is the coefficient of the adaption for the controller gain and positive, i.e., $\gamma_{k}>0$.

\subsection{Type-2 Neuro-Fuzzy Controller}
An interval type-2 Takagi-Sugeno-Kang fuzzy \emph{if-then} rule $R_{ij}$ is employed in this study \cite{Wenchuan2017,Zheng2018}. The premise parts are type-2 fuzzy functions while the consequent parts are crisp numbers. The rule $R_{ij}$ is defined as below:
\begin{equation}\label{Rlinearfunction}
R_{ij}: \;\; \textrm{If} \; e \; \textrm{is} \;\; \widetilde{1}_i \;\; \textrm{and} \; \dot{e} \; \textrm{is} \;\; \widetilde{2}_j, \;\; \textrm{then} \; f_{ij}=d_{ij}
\end{equation}
where $e$ and $\dot{e}$ denote respectively the error and the error rate while $\widetilde{1}_i$ and $\widetilde{2}_j$ denote respectively type-2 fuzzy sets for the input 1 and input 2. The function $f_{ij}$ is the output of the rules and the total number of rules is equal to $K=I \times J$ in which $I$ and $J$ denote respectively the numbers of membership functions used for the input 1 and input 2. The detailed of the proposed T2NFS is shown in Fig. \ref{fig_t2nfs}.
\begin{figure*}[b!]
  \centering
  \includegraphics[width=5.0in]{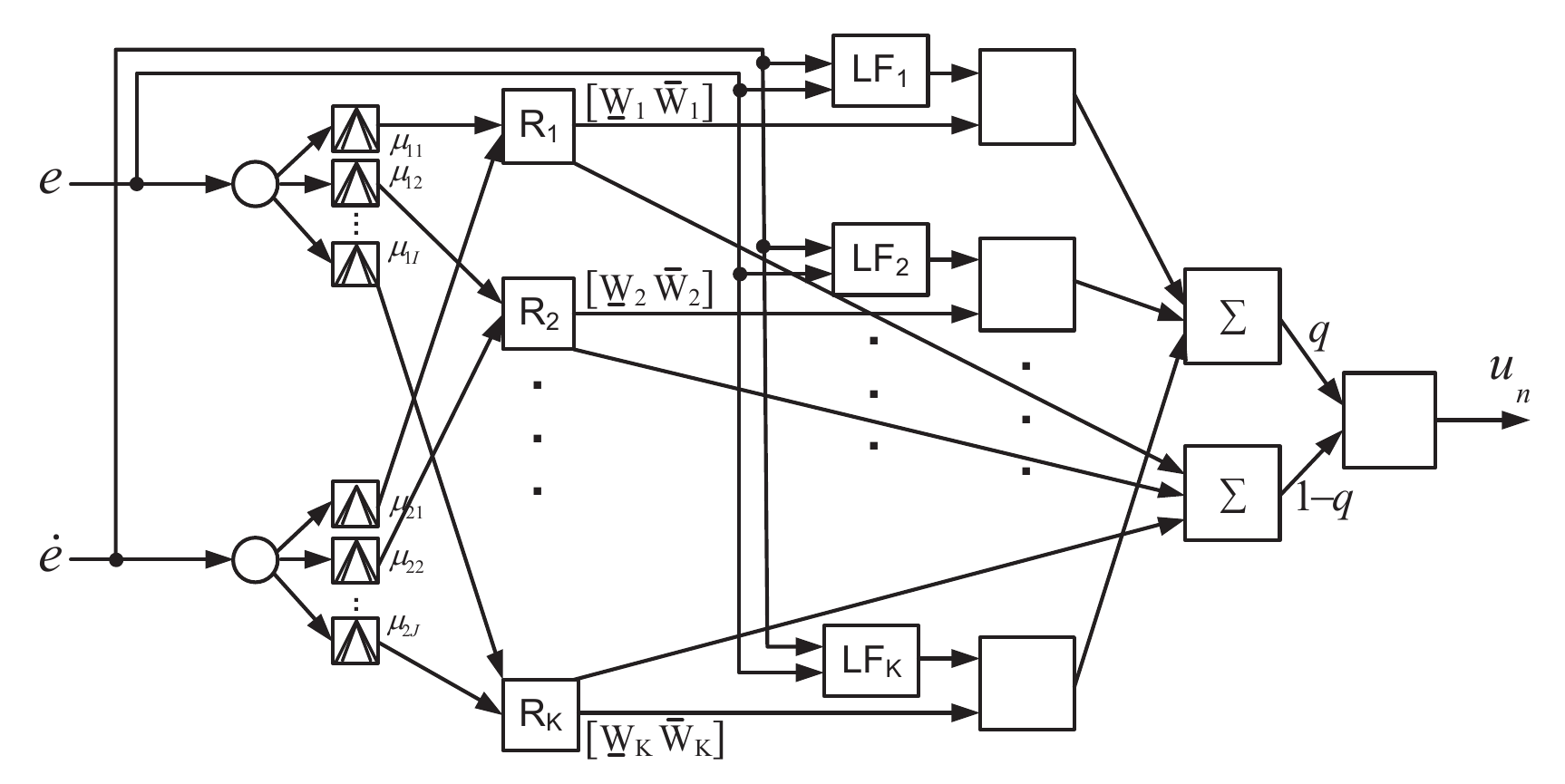}\\
  \caption{Structure of the proposed T2NFS for two inputs}\label{fig_t2nfs}
\end{figure*}

The lower and upper Gaussian MFs for T2FLSs are written as follows:
\begin{equation}\label{mu1_lower}
\underline{\mu}_{1i}(e) = \exp\Bigg(-\bigg(\frac{e-\underline{c}_{1i}}{\underline{\sigma}_{1i}}\bigg)^2\Bigg)
\end{equation}
\begin{equation}\label{mu1_upper}
\overline{\mu}_{1i}(e) = \exp\Bigg(-\bigg(\frac{e-\overline{c}_{1i}}{\overline{\sigma}_{1i}}\bigg)^2\Bigg)
\end{equation}
\begin{equation}\label{mu2_lower}
\underline{\mu}_{2j}(\dot{e}) = \exp\Bigg(-\bigg(\frac{\dot{e}-\underline{c}_{2j}}{\underline{\sigma}_{2j}}\bigg)^2\Bigg)
\end{equation}
\begin{equation}\label{mu2_upper}
\overline{\mu}_{2j}(\dot{e}) = \exp\Bigg(-\bigg(\frac{\dot{e}-\overline{c}_{2j}}{\overline{\sigma}_{2j}}\bigg)^2\Bigg)
\end{equation}
where $\underline{c}$ and $\overline{c}$ denote respectively the lower and upper means the membership functions, and $\underline{\sigma}$ and $ \overline{\sigma}$ the lower and upper standard deviations of the  membership functions. They are adjusted for the T2NFC method throughout the simulations.

The lower and upper memberships $\underline{\mu }$ and $\overline{\mu }$ of fuzzy system are written for each input signal being fed to the system. The firing strengths of rules are formulated as follows:
\begin{equation}\label{wij_lower_upper}
\underline{w}_{ij} = \underline{\mu}_{1i}(e)  \underline{\mu}_{2j}(\dot{e}) \quad \textrm{and} \quad
\overline{w}_{ij} = \overline{\mu}_{1i}(e) \overline{\mu}_{2j}(\dot{e})
\end{equation}

The control signal generated by the network is formulated below \cite{Biglarbegian2010}:
\begin{equation}\label{un}
u_n=q \sum_{i=1}^{I}\sum_{j=1}^{J}f_{ij}\widetilde{\underline{w}}_{ij}+(1-q)\sum_{i=1}^{I}\sum_{j=1}^{J}f_{ij}\widetilde{\overline{w}}_{ij}
\end{equation}
where $\widetilde{\underline{{w}}}_{ij}$ and $\widetilde{\overline{{w}}}_{ij}$ are the normalized firing strengths of the lower and upper output signals of the neuron $ij$ and written as follows:
\begin{eqnarray}\label{wij_lower_upper_normalized}
\widetilde{\underline{w}}_{ij} = \frac{\underline{w}_{ij}}{\sum_{i=1}^{I}\sum_{j=1}^{J}\underline{w}_{ij}} \quad\textrm{and} \quad
\widetilde{\overline{w}}_{ij} = \frac{\overline{w}_{ij}}{\sum_{i=1}^{I}\sum_{j=1}^{J}\overline{w}_{ij}}
\end{eqnarray}
The parameter $q$ determines the contributions of the lower and upper firing levels to the control signal and the adaptation rule for the parameter $q$ is given Section \ref{sec_smla}.

\subsection{Sliding Mode Learning Algorithm}\label{sec_smla}
The sliding surface utilized the fundamental basis of SMC theory \cite{Utkin} is written as follows:
\begin{equation}\label{eq_s}
s = \Big(\frac{d }{d t} + \lambda \Big)^{n-1} e
\end{equation}
where $\lambda$, $e$ and $n$ denote respectively the slope of the sliding surface, the system error and the order of the system to be controlled. The slope of the sliding surface is a strictly positive constant, i.e., $\lambda>0$. The error is the difference between a desired setpoint $x_{d}$ and a measured state $x$, i.e., $e= x_{d} - x$. The sliding surface is utilized to train the sliding mode learning algorithm. The T2NFC parameters are updated with the following equations:
\begin{equation} \label{c_1i_lower}
\dot{\underline{c}}_{1i} = \dot{e} + (e - \underline{c}_{1i}) \alpha \textrm{sgn}\left(s \right)
\end{equation}
\begin{equation} \label{c_1i_upper}
\dot{\overline{c}}_{1i} = \dot{e} + (e - \overline{c}_{1i}) \alpha \textrm{sgn}\left(s \right)
\end{equation}
\begin{equation} \label{c_2j_lower}
\dot{\underline{c}}_{2j} = \ddot{e} + (\dot{e} - \underline{c}_{2j}) \alpha \textrm{sgn}\left(s \right)
\end{equation}
\begin{equation} \label{c_2j_upper}
\dot{\overline{c}}_{2j} = \ddot{e} + (\dot{e} - \overline{c}_{2j}) \alpha \textrm{sgn}\left(s \right)
\end{equation}
\begin{equation}\label{sigma_1i_lower}
\dot{\underline{\sigma}}_{1i} = - \underline{\sigma}_{1i}  \bigg( 1+ \Big(\frac{ \underline{\sigma}_{1i} }{e - \underline{c}_{1i}} \Big)^{2} \bigg) \alpha  \textrm{sgn}\left(s \right)
\end{equation}
\begin{equation}\label{sigma_1i_upper}
\dot{\overline{\sigma}}_{1i} = - \overline{\sigma}_{1i}  \bigg( 1 + \Big( \frac{ \overline{\sigma}_{1i} }{ e - \overline{c}_{1i} } \Big)^{2} \bigg) \alpha  \textrm{sgn}\left(s \right)
\end{equation}
\begin{equation}\label{sigma_2j_lower}
\dot{\underline{\sigma}}_{2j} = - \underline{\sigma}_{2j}  \bigg( 1 + \Big( \frac{ \underline{\sigma}_{2j} }{ \dot{e} - \underline{c}_{2j}} \Big)^{2} \bigg) \alpha  \textrm{sgn}\left(s \right)
\end{equation}
\begin{equation}\label{sigma_2j_upper}
\dot{\overline{\sigma}}_{2j} = - \overline{\sigma}_{2j} \bigg( 1 + \Big( \frac{ \overline{\sigma}_{2j} }{ \dot{e}  - \overline{c}_{2j} } \Big)^{2} \bigg) \alpha  \textrm{sgn}\left(s \right)
\end{equation}
\begin{equation}\label{f_ij}
\dot{f}_{ij} =\frac{q \widetilde{\underline{w}}_{ij}+ (1-q)\widetilde{\overline{w}}_{ij}}{(q\widetilde{\underline{W}}+(1-q) \widetilde{\overline{W}})^T(q\widetilde{\underline{W}}+ (1-q)\widetilde{\overline{W}})}\alpha   \textrm{sgn}\left(s \right)
\end{equation}
\begin{equation}\label{q}
\dot{q} =\frac{1}{F(\widetilde{\underline{W}}-\widetilde{\overline{W}})^{T}}\alpha   \textrm{sgn}\left(s \right)
\end{equation}
\begin{equation}\label{alpha}
\dot{\alpha} =\gamma_{\alpha} \mid s \mid
\end{equation}
where $\alpha$ is the learning rate while $\gamma_{\alpha}$ is the coefficient to update the learning rate, and they are positive, i.e., $\alpha, \gamma_{\alpha}>0$. Normalized firing strengths vectors $\widetilde{\underline{W}}(t)$ and $\widetilde{\overline{W}}(t) $,  and the consequent part vector $F$ are defined as follows:
\begin{eqnarray}
\widetilde{\underline{W}}&=& [\widetilde{\underline{w}}_{11} \; \widetilde{\underline{w}}_{12} \dots \widetilde{\underline{w}}_{21} \dots \widetilde{\underline{w}}_{ij} \dots \widetilde{\underline{w}}_{IJ}]^{T} \nonumber \\
\widetilde{\overline{W}} &=& [\widetilde{\overline{w}}_{11} \; \widetilde{\overline{w}}_{12} \dots \widetilde{\overline{w}}_{21} \dots \widetilde{\overline{w}}_{ij} \dots \widetilde{\overline{w}}_{IJ}]^{T} \nonumber \\
F &=& [f_{11} \; f_{12} \dots f_{21} \dots f_{ij} \dots f_{IJ}] \nonumber
\end{eqnarray}
where $0<\widetilde{\underline{w}}_{ij} \leq 1$ and $0<\widetilde{\overline{w}}_{ij} \leq 1$. Moreover $\sum_{i=1}^{I}\sum_{j=1}^{J}\widetilde{\underline{w}}_{ij} = 1$ and $\sum _{i=1}^{I}\sum_{j=1}^{J}\widetilde{\overline{w}}_{ij} = 1$.

\begin{assumption}\label{ass_BuBdotu}
The input signal and its time derivative are assumed to be bounded in a compact set:
\begin{equation}\label{eq_BuBdotu}
\mid u \mid < B_{u}, \qquad \mid \dot{u} \mid < B_{\dot{u}}
\end{equation}
where $B_{u}$ and $B_{\dot{u}}$ are considered as positive constants.
\end{assumption}

\begin{theorem}[Stability of sliding mode learning algorithm]\label{theorem1}
If adaptations rules are proposed as in \eqref{c_1i_lower}-\eqref{alpha}, the final value of the learning rate $\alpha^{*}$ is large enough, i.e., $\alpha^{*} > B_{\dot{u}}$, and the term, which is the product of the controller gain $k$ and its final value $k^*$, is larger than the absolute value of the sliding surface $\mid s \mid$, i.e., $ k^* k > \mid s \mid$, this ensures that the output of the conventional controller $u_{c}$ will converge to zero in finite time.  
\end{theorem}

\begin{proof}
The stability of the learning algorithm is checked by using the following Lyapunov candidate function:
\begin{equation}\label{eq_smc_Lyapunov}
V= \frac{1}{2k^{*}} u_{c}^{2} + \frac{(k^{*})^{2}}{2 \gamma_{\alpha} } \Big(\frac{\alpha}{(k^{*})^{2}} - \alpha^{*}  \Big)^{2}
\end{equation}
The time-derivative of the aforementioned Lyapunov candidate function is calculated as follows:
\begin{equation}\label{eq_smc_Lyapunov_d} 
\dot{V} =   \frac{u_{c} \dot{u}_{c}}{k^{*}} + \frac{\dot{\alpha}}{  \gamma_{\alpha}} \Big(\frac{\alpha}{(k^{*})^{2}}  - \alpha^{*} \Big) 
\end{equation}
If \eqref{eq_totalcontrolaction} and \eqref{alpha} are inserted into \eqref{eq_smc_Lyapunov_d}, it is obtained as:
\begin{equation}\label{eq_smc_Lyapunov_dd} 
\dot{V} =   \frac{u_{c}}{k^{*}} (-\dot{u}_{n} + \dot{u} ) + \mid s \mid \Big( \frac{\alpha}{(k^{*})^{2}}  - \alpha^{*} \Big)
\end{equation}
$\dot{u}_{n}$ in \eqref{dotVc4} is inserted into \eqref{eq_smc_Lyapunov_dd}, the time-derivative of the Lyapunov function is written as:
\begin{equation}\label{eq_smc_Lyapunov_d2}
\dot{V} =   \frac{u_{c}}{k^{*}} \Big( -2\alpha  \textrm{sgn}\left( s \right) + \dot{u} \Big) + \mid s \mid \Big( \frac{\alpha}{(k^{*})^{2}}  - \alpha^{*} \Big)
\end{equation}
As it is stated in Assumption \ref{ass_BuBdotu}, the total input rate $\dot{u}$ is upper bounded by $B_{\dot{u}}$, and the conventional control action in \eqref{eq_robustterm} is inserted into \eqref{eq_smc_Lyapunov_d2}
\begin{eqnarray}\label{eq_smc_Lyapunov_d3}
\dot{V} & < &  \frac{- 2 \alpha k}{k^*} + \frac{k}{k^*} \mid s \mid B_{\dot{u}}  + \frac{\alpha \mid s\mid}{(k^{*})^{2}} -\alpha^{*} \mid s \mid  \nonumber \\
\dot{V} & < & \alpha \underbrace{\Big( \frac{-2 k k^* + \mid s \mid }{(k^*)^2} \Big)}_{<0} + \mid s \mid \underbrace{ \Big(\frac{k}{k^*}
 B_{\dot{u}}   -  \alpha^{*} \Big)}_{<0}
\end{eqnarray}
where $\alpha^{*}$ and $k^{*}$ are the final values of the learning rate and controller gain, and considered to be unknown parameters that are determined throughout the adaptations of learning rate and controller gain. If the final value of the learning rate $\alpha^{*}$ is large enough, i.e., $\alpha^{*} > B_{\dot{u}}$, and the term, which is the product of the controller gain $k$ and its final value $k^*$, is larger than the absolute value of the sliding surface $\mid s \mid$, i.e., $k^* k > \mid s \mid$, then the time-derivative of the Lyapunov candidate function is negative, i.e., $\dot{V}<0$. This purports that the sliding mode learning algorithm is stable. In other words, the output of the conventional controller $u_{c}$ converges to zero after the T2NFC algorithm learns the system. 
\end{proof}

\begin{remark}
As can be seen in \eqref{eq_robustterm_k}, the controller gain only increases throughout simulations.  Therefore, the final value of the controller $k^{*}$ is equal or bigger than $k$.  This results in $\frac{k}{k^*} \leq 1$.
\end{remark}

\begin{remark}
The adaptation rules for the conventional controller gain in \eqref{eq_robustterm_k}  and learning rate in \eqref{alpha} are enforced; therefore, the final values of the controller gain and learning rate are unknown and are determined throughout the adaptations of the controller gain and learning rate, and they can reach at large values to make system and sliding mode learning algorithm stable. 
\end{remark}

\section{Stability Analysis of the Overall System}\label{sec_overalstability}

An nth-order uncertain nonlinear system is described in the following form:
\begin{eqnarray}\label{eq_nonlinearsystem}
\dot{x}_{1} & = & x_{2} \nonumber \\
\vdots & \vdots & \vdots \nonumber \\
\dot{x}_{n-1} & = & x_{n} \nonumber \\
\dot{x}_{n} & = & f(\textbf{x}) + g u + d
\end{eqnarray}
where $f(\textbf{x})$ is the nonlinear function, $g$ is the coefficient of the control input and positive, $u$ is the control input, $d$ is the disturbance, and $\textbf{x}=[x_{1}, x_{2}, \hdots, x_{n}]^{T}$ is the state vector. The error $e$ is defined as $e=x_{d}-x_{1}$ where $x_{d}$ is the desired state. 
\begin{assumption}\label{ass_}
The nth-order time-derivative of the desired state, the nonlinear function, the disturbance, the output of the T2NFC and the term $E$, which is equal to $\dot{s}-\frac{d^{n} e}{d t^{n}}$, are assumed to be bounded:
\begin{equation}\label{eq_bounds}
\mid \frac{d^{n} x_{d}}{d t^{n}} \mid + \mid f(\textbf{x}) \mid + \mid d \mid + g \mid u_{n} \mid + \mid E \mid < B
\end{equation}
where $B$ is considered as a positive constant.
\end{assumption}
 
\begin{theorem}[Stability of the overall system]\label{theorem2}
The system to be controlled is the form of \eqref{eq_nonlinearsystem} and the controller in \eqref{eq_totalcontrolaction} is applied to the system. Then, if the final value of the controller gain $k^{*}$ is large enough, i.e., $k^{*} > 2 B$, the sliding surface $s$ converges to zero for an arbitrary initial condition. 
\end{theorem}

\begin{proof}
The stability of the overall system is checked for an nth-order uncertain nonlinear system by following Lyapunov candidate function:
\begin{equation}\label{eq_lyapunov}
V = \frac{1}{2} s^{2}  + \frac{g}{2 \gamma_{k}} (k - k^{*})^{2}
\end{equation}
The time-derivative of $V$ is calculated as: 
\begin{equation}\label{eq_lyapunovrate}
\dot{V} =   s \dot{s}  + \frac{g \dot{k}}{\gamma_{k}} (k - k^{*})
\end{equation}
The sliding surface rate $\dot{s}$ is obtained by taking time-derivative of the sliding surface defined in \eqref{eq_s}. Then, it is inserted into \eqref{eq_lyapunovrate} as follows:
\begin{equation}\label{eq_lyapunovratee}
\dot{V}  =  s \Big(\frac{d^{n} e}{d t^{n}} \underbrace {+ \ldots+ \lambda^{n-1} \dot{e} \Big)}_{E} + \frac{g \dot{k}}{\gamma_{k}} (k - k^{*})
\end{equation}
Since $E=\dot{s}-\frac{d^{n} e}{d t^{n}}$ as stated in Assumption \ref{ass_} and $\frac{d^{n} e}{d t^{n}} = \frac{d^{n} x_{d}}{d t^{n}} - \frac{d^{n} x}{d t^{n}}$ , it is obtained as 
\begin{eqnarray}\label{eq_lyapunovratee2}
\dot{V}  = s \Big(\frac{d^{n} x_{d}}{d t^{n}} - \frac{d^{n} x}{d t^{n}} +  E \Big) + \frac{g\dot{k}}{\gamma_{k}} (k - k^{*}) 
\end{eqnarray}
Then, it is re-written by taking nonlinear system in \eqref{eq_nonlinearsystem} into consideration as follows:
\begin{eqnarray}\label{eq_lyapunovratee3}
\dot{V}  =  s \Big(  \frac{d^{n} x_{d}}{d t^{n}} - f(\textbf{x})- d - g u +  E \Big) + \frac{g\dot{k}}{\gamma_{k}} (k - k^{*} )
\end{eqnarray}
The control law in \eqref{eq_totalcontrolaction} and the adaptation for the controller gain in \eqref{eq_robustterm_k} are inserted into \eqref{eq_lyapunovratee3}.
\begin{equation}\label{eq_lyapunovrate2}
\dot{V}  =  s \Big(  \frac{d^{n} x_{d}}{d t^{n}} -  f(\textbf{x})  - d  - g u_{n} - g u_{c} + E \Big) + \frac{g \mid s \mid}{2} (k - k^{*})
\end{equation}
The conventional control action in \eqref{eq_robustterm} is inserted into \eqref{eq_lyapunovrate2}
\begin{equation}\label{eq_lyapunovrate3}
\dot{V}  =  s \Big(  \frac{d^{n} x_{d}}{d t^{n}} -  f(\textbf{x})  - d  - g u_{n} + E  - g k  \textrm{sgn}\left(s \right)  \Big) + \frac{g \mid s \mid}{2} (k - k^{*})
\end{equation}
As stated in Assumption \ref{ass_}, the term $E$, the nth-order time-derivative of the reference $\frac{d^{n} x_{d}}{d t^{n}} $, the function $f(\textbf{x})$, the disturbance $d$ and the output of the T2NFC $u_{n}$ are assumed to be upper bounded, the aforementioned equation \eqref{eq_lyapunovrate3} is obtained as follows:
\begin{equation}\label{eq_lyapunovrate5}
 \dot{V} < \mid s \mid ( B - g k ) + \frac{g \mid s \mid}{2} (k - k^{*})
\end{equation}
where $k^{*}$ is the final value of the controller gain and considered to be an unknown parameter that is determined throughout the adaptation of the controller gain. If the final value of the controller gain $k^{*}$ is large enough, i.e., $k^{*} > 2 B $, $\dot{V}$ yields
\begin{equation}
 \dot{V}< -\frac{g k \mid s \mid}{2}
 \end{equation}
Since $g, k >0$, this guarantees the stability of the overall system and purports that the sliding surface converges to zero in finite time.
\end{proof}

\begin{remark}
Since the controller gain can reach at a large value to make the system stable, it is not required to know the upper bound. This is one of the advantages of the SMLC algorithm.
\end{remark}

\begin{remark}
A sliding motion appears on the sliding manifold $s=0$ after finite time $t_{h}$ if the condition $s \dot{s} <0$ is satisfied for all $t$ in some nontrivial semi-open sub-interval of time of the form $[t, t_{h})\subset (-\infty, t_{h})$. Therefore, a sliding motion will be maintained on $s=0$ for all $t> t_{h}$.
\end{remark}

\section{Simulation Results}\label{section_simulation}

The performance of the developed SMLC algorithm against an external disturbance is firstly shown in Section \ref{section_pendulum} and then analyzed under noisy conditions in Section \ref{section_num_ex}. Throughout simulation studies, the number of membership functions of the T2NFC are set to 3, i.e., $I = J = 3$ and the sampling time is set to $0.01$ second, Since SMC theory suffers from high-frequency oscillations, the following function in \eqref{chatter} is used as the sign function to remove the chattering effect:
\begin{equation}\label{chatter}
 \textrm{sgn}\left(s \right):=\frac{s}{|s|+\chi}
\end{equation}
with $\chi=0.05$.

To avoid division by zero in the adaptation laws of
\eqref{c_1i_lower}-\eqref{q}, an instruction is included in the algorithm to make the denominator equal to 0.001 when its calculated value is smaller than this threshold. Moreover, the controller gain of the conventional controller and learning rate of the learning algorithm must not have infinite values under noisy conditions; therefore, a dead-zone is employed. If the sliding surface is within the bounds of the dead-zone parameter $\epsilon=0.001$, i.e., $ \mid s \mid < \epsilon$, then these parameters are not updated. Furthermore, the reference must be changing smoothly and there must exist no discontinuous disturbance to avoid infinite values for the controller gain and learning rate.

\subsection{Scenerio 1: Control Performance Against External Disturbance}\label{section_pendulum}

To evaluate control performance against disturbances, the developed SMLC algorithm is applied to an adaptive cruise control system, which is illustrated in Fig. \ref{fig_ACC_schematic} and the control structure is shown in Fig. \ref{fig_accsystem}. The equations of motion for the longitudinal vehicle dynamics are written as follows \cite{Stankovic2000,Kayacan2017}

\begin{eqnarray}\label{}
\dot{x}_{1} & = & x_{2} \nonumber \\
\dot{x}_{2} & = & x_{3} \nonumber \\
\dot{x}_{3} & = & -2\frac{k_{a}}{m} x_{2} x_{3} - \frac{1}{\tau} \Big[x_{3} + \frac{k_{a}}{m} x_{2}^{2}  \Big] + \frac{u}{m \tau} + d
\end{eqnarray}
where $x_{1}$, $x_{2}$, $x_{3}$, m, $k_{a}$, $\tau$, u and d represent respectively the position, velocity and acceleration, the mass of the vehicle, the aerodynamic drag coefficient, the engine time-lag, the throttle command input and the external disturbance. The numerical values for the parameters in this study are $m=9$ $kg$, $k_{a}=0.26$, $\tau=0.1$ \cite{Mehra2015}.

\begin{figure}[h!]
  \centering
  \includegraphics[width=4.0in]{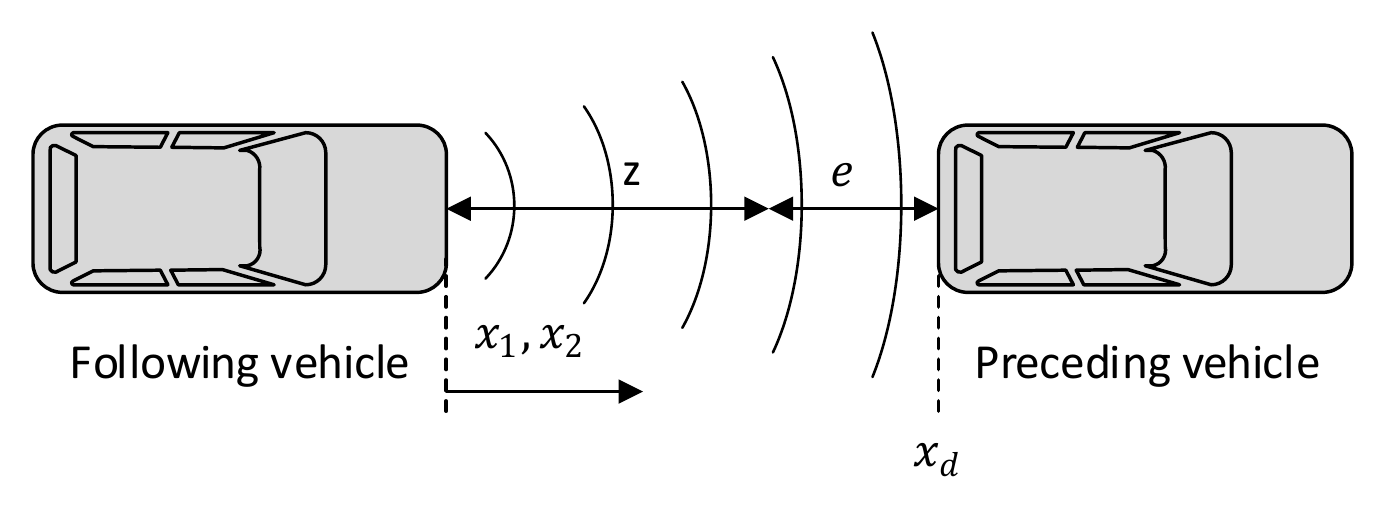}\\
  \caption{Schematic diagram of the adaptive cruise control system}\label{fig_ACC_schematic}
\end{figure}

\begin{figure}[h!]
  \centering
  \includegraphics[width=4.0in]{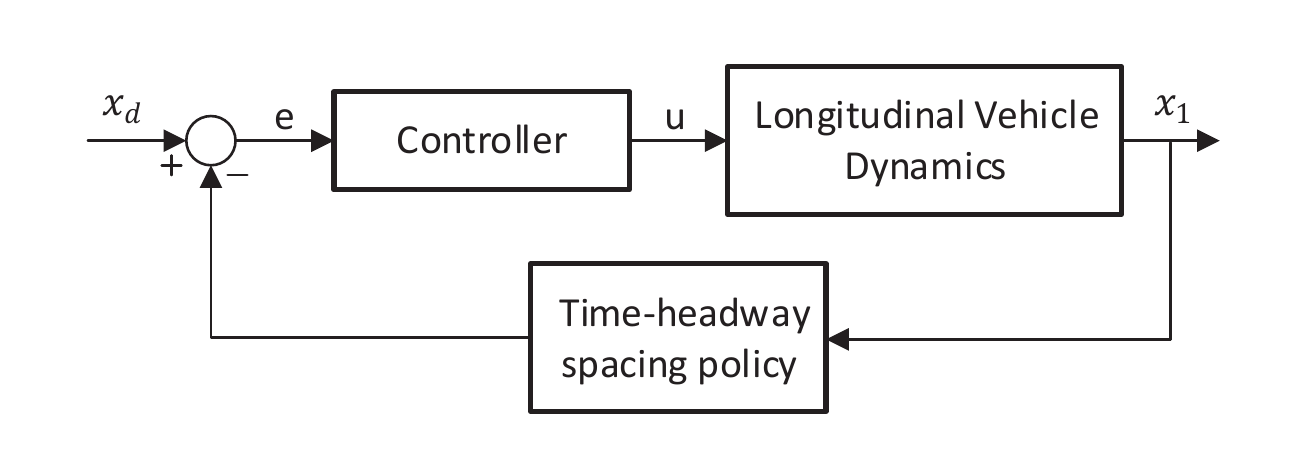}\\
  \caption{Control structure of the adaptive cruise control system}\label{fig_accsystem}
\end{figure}

The time-headway spacing policy to follow the preceding vehicle with the desired relative distance is  formulated as follows:
\begin{equation}\label{eq_timeheadway}
z=h x_{2}
\end{equation}
where z is the desired relative inter-vehicle distance and h is the desired time headway. The spacing policy is called the constant time-headway time policy which aims a constant inter-vehicle time gap. The position reference in this study is considered as the position of the preceding vehicle. The relative spacing error is defined by taking the spacing policy in \eqref{eq_timeheadway} into account as follows:
\begin{equation}\label{eq_error}
e=x_{d} - x_{1} - h x_{2}
\end{equation}

The sliding surface for the adaptive cruise control system is second-order since the adaptive cruise control systems is a third-order system. The slope of the sliding surface $\lambda$ to update the T2NFC is set to $1$; therefore, the sliding surface becomes as $s=\ddot{e} + 2 \dot{e} + e$ as formulated in \eqref{eq_s}. The coefficients $\gamma_{k}$ and $\gamma_{\alpha}$ for the adaptions of the controller gain in  \eqref{eq_robustterm_k} and the learning rate in \eqref{alpha} are set to $0.1$. The initial conditions on the states of the system are considered as $[x_{1}, x_{2}, , x_{3}]=[0,0, 0]^{T}$. An external disturbance is imposed in the vehicle dynamics as follows:
\begin{equation}\label{eq_disturbance}
  d(t) = 1+ 0.25 \sin{(t)}  
  \end{equation}

The test trajectory is defined by the position reference $x_{d}(t)$ consists of ramp signals as follows:
\begin{equation}\label{eq_aref}
  x_{d}(t) = \Bigg \{
    \begin{array}{rl}
    0 \leq t < 20  & x_{d}= t \\
   20 \leq t < 40 & x_{d}= t + \frac{0.05}{2}(t-20)^{2} \\
   40 \leq t < 40 & x_{d}= t + \frac{0.05}{2}(t-20)^{2} - \frac{0.05}{2}(t-40)^{2}  
 \end{array} 
\end{equation}

The position, speed and acceleration responses of the vehicle are respectively shown in Figs. \ref{fig_pos}-\ref{fig_acc}. The developed SMLC algorithm can learn the system behavior online in finite time so that it exhibits robust control performance against the external disturbance. As the error is shown in Fig. \ref{fig_error}, it controls the system without any steady-state error. 

The control signals are shown in Fig. \ref{fig_controlsignals}. The generated control signal by the conventional control action $u_c$ converges zero while the generated control signal by the T2NFC $u_n$ takes overall control signal. Thus, the T2NFC becomes the leading controller after a short time period. The output of the conventional controller $u_{c}$ becomes non-zero only at the beginning of the simulations while the T2NFC is learning system behaviour.

The adaptations of the controller gain $k$, the learning rate $\alpha$ and the parameter $q$ are presented in Figs. \ref{fig_k}-\ref{fig_q}. The parameter $q$ is generally set to $0.5$ in T2FLCs while the controller gain and learning rate are chosen arbitrarily in this study. Therefore, the initial conditions of these parameters are set to $k(0)=1$, $\alpha(0)=3$ and $q(0)=0.5$. Thanks to these adaptations, the controller gain, the learning rate and the proportions of the upper and lower membership functions are updated throughout the simulations. The adaptations of the controller gain $k$ and learning rate $\alpha$ converge to a certain value since (i) the reference is changing smoothly, (ii) there is no discontinuous disturbance on the system and (iii) the dead-zone parameter is not lower than noise-level in real-time applications.

In literature, disturbance observers-based control approaches have been generally implemented for uncertain systems \cite{7569085}. In these approaches, unmodeled dynamics are considered as an external disturbance, and a control law including disturbance estimate is derived. In our work, there is no need of such an observer due to the fact that the proposed algorithm learns system behavior. This is one of the superiorities of the developed algorithm in the paper.  

\begin{figure*}[t!]
\centering
\subfigure[ ]{
\includegraphics[width=0.47\textwidth]{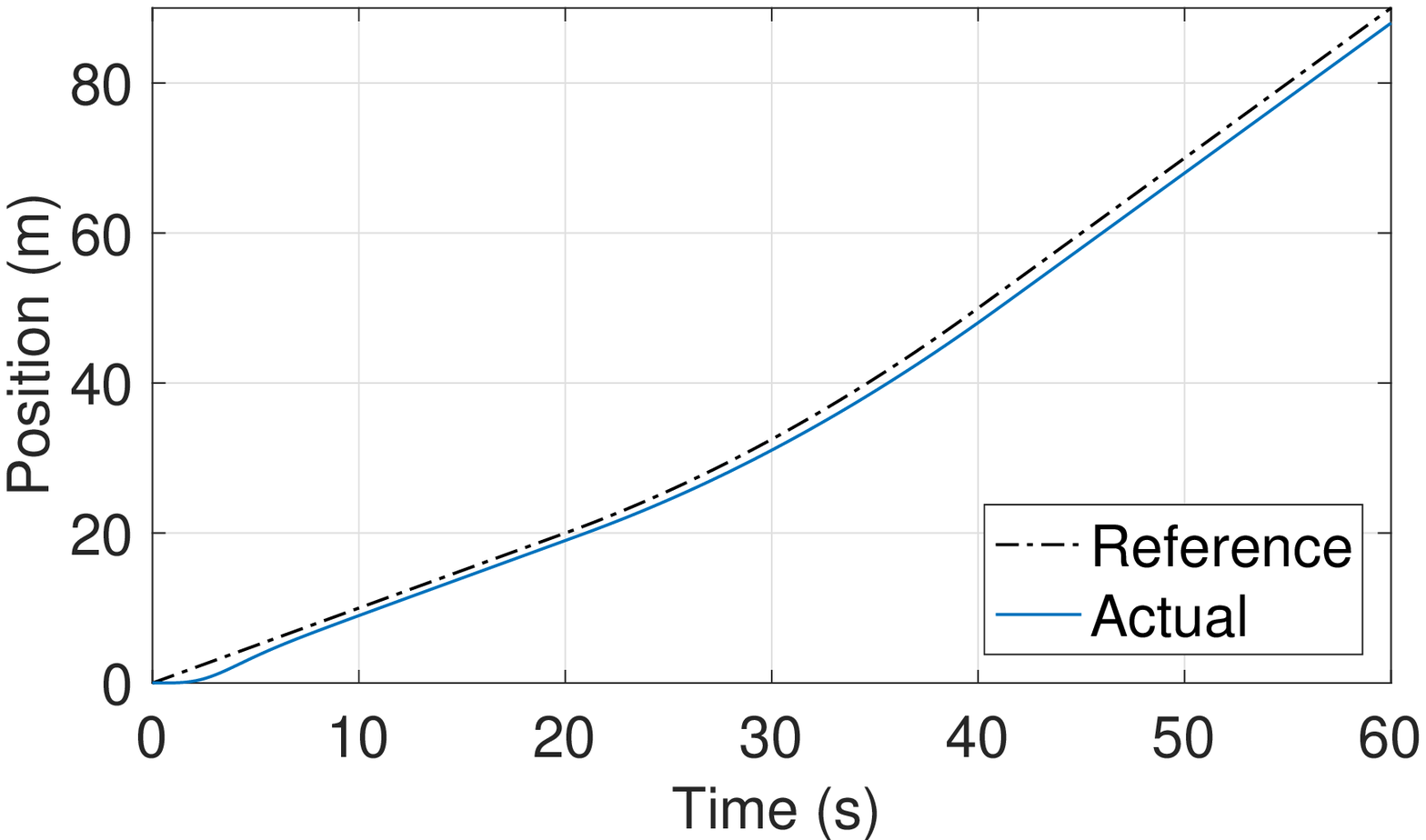}
\label{fig_pos}
}
\subfigure[ ]{
\includegraphics[width=0.47\textwidth]{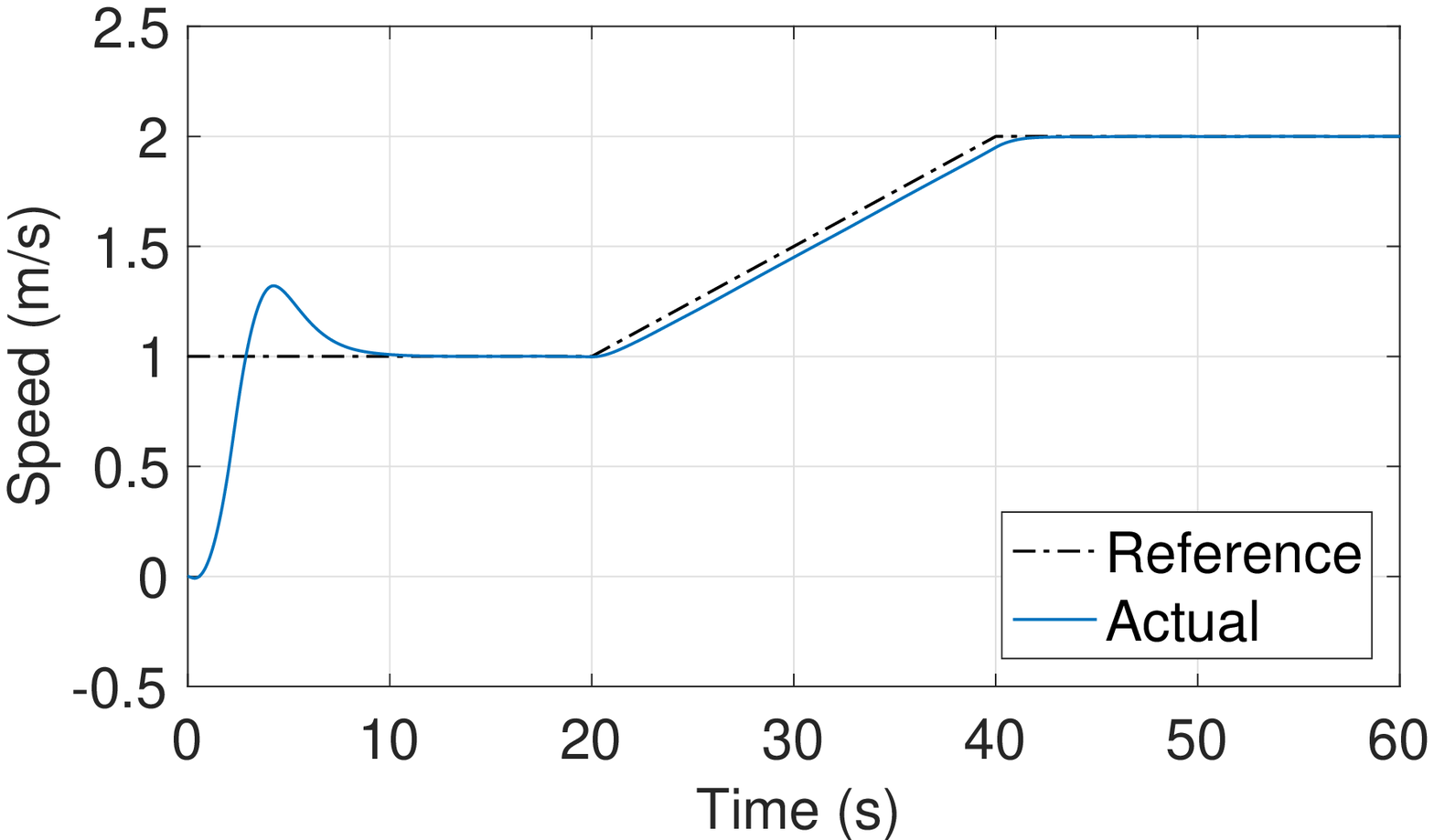}
\label{fig_speed}
}
\subfigure[ ]{
\includegraphics[width=0.47\textwidth]{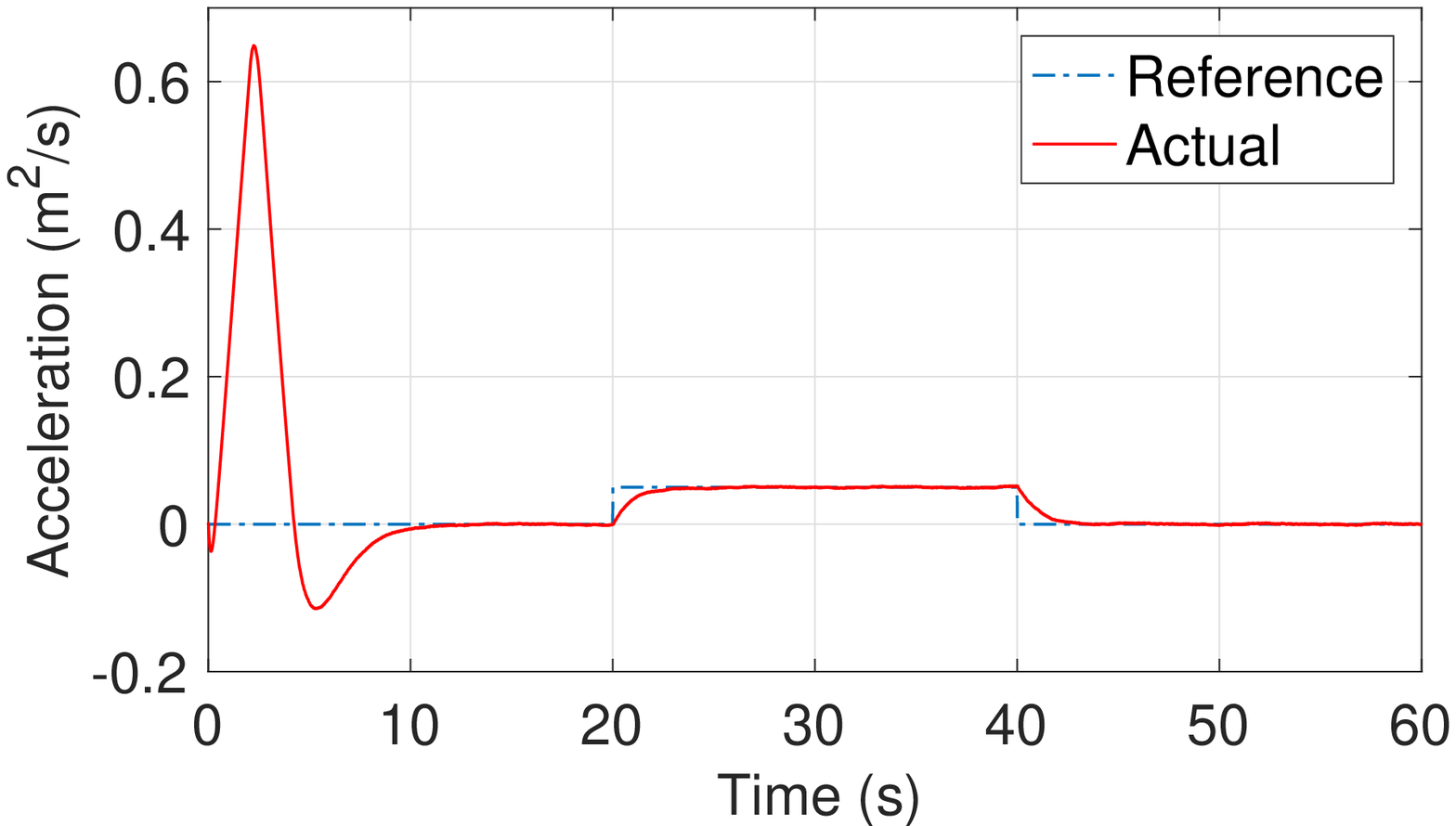}
\label{fig_acc}
}
\subfigure[ ]{
\includegraphics[width=0.47\textwidth]{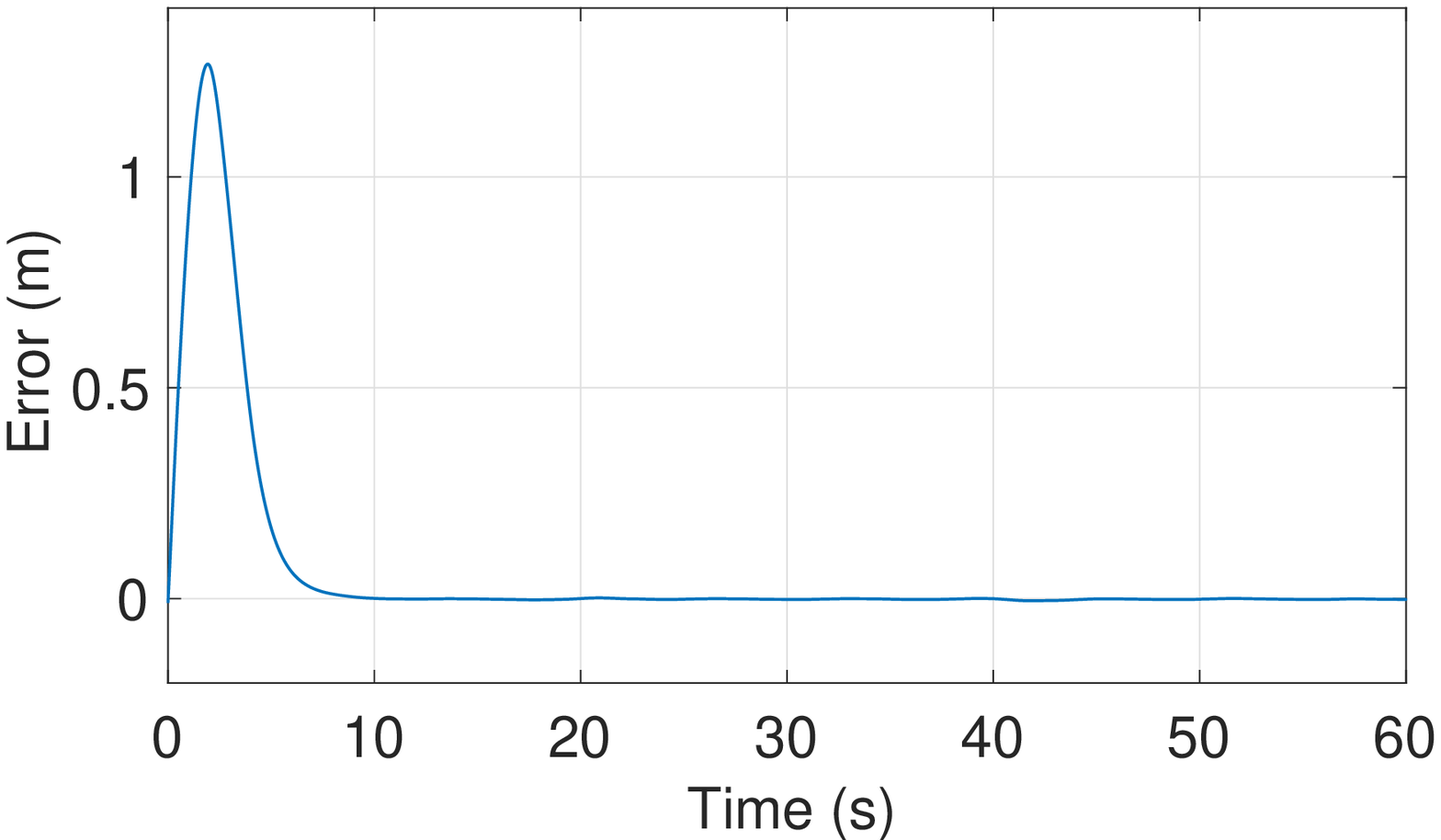}
\label{fig_error}
}
\subfigure[ ]{
\includegraphics[width=0.47\textwidth]{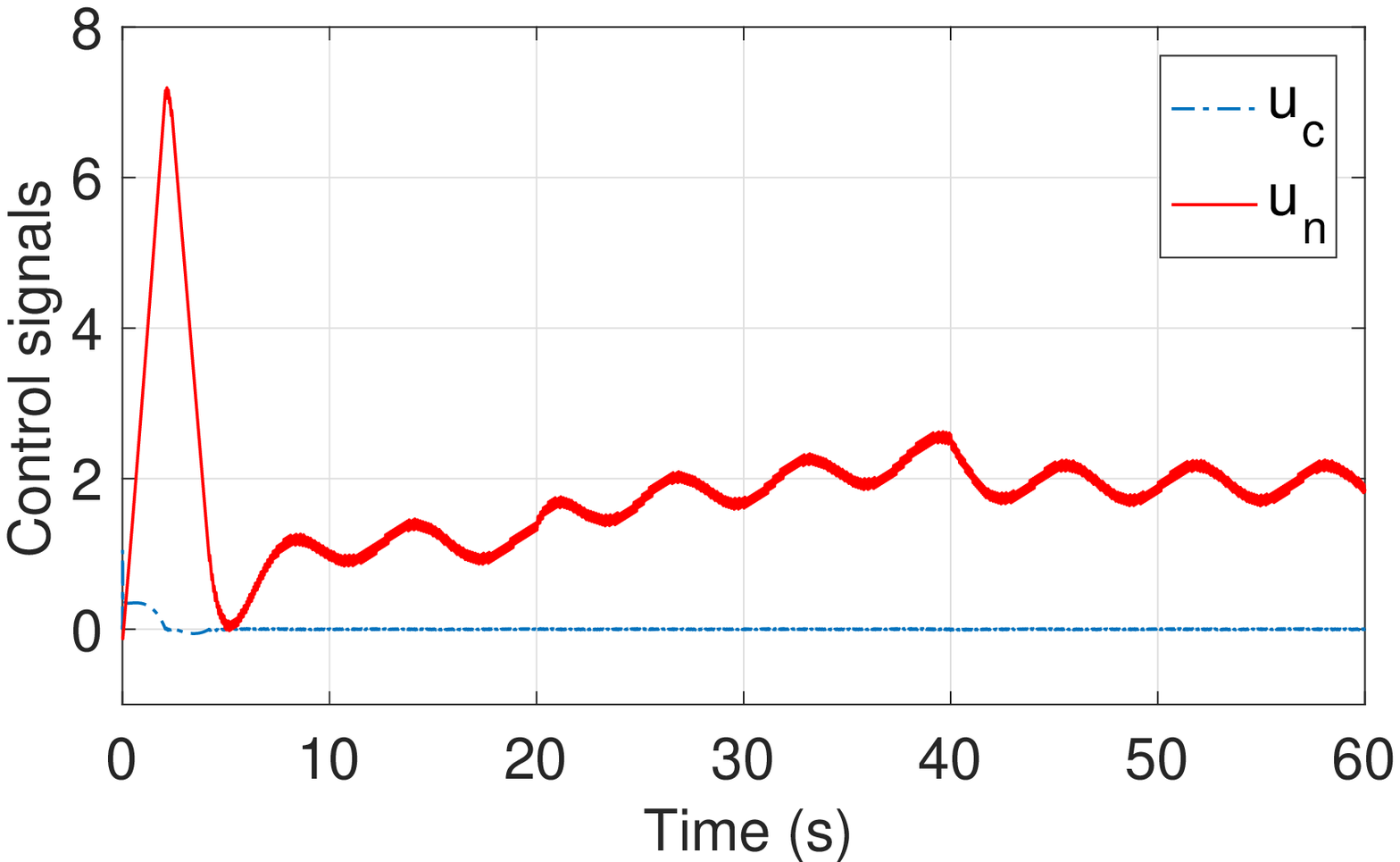}
\label{fig_controlsignals}
}
\subfigure[ ]{
\includegraphics[width=0.47\textwidth]{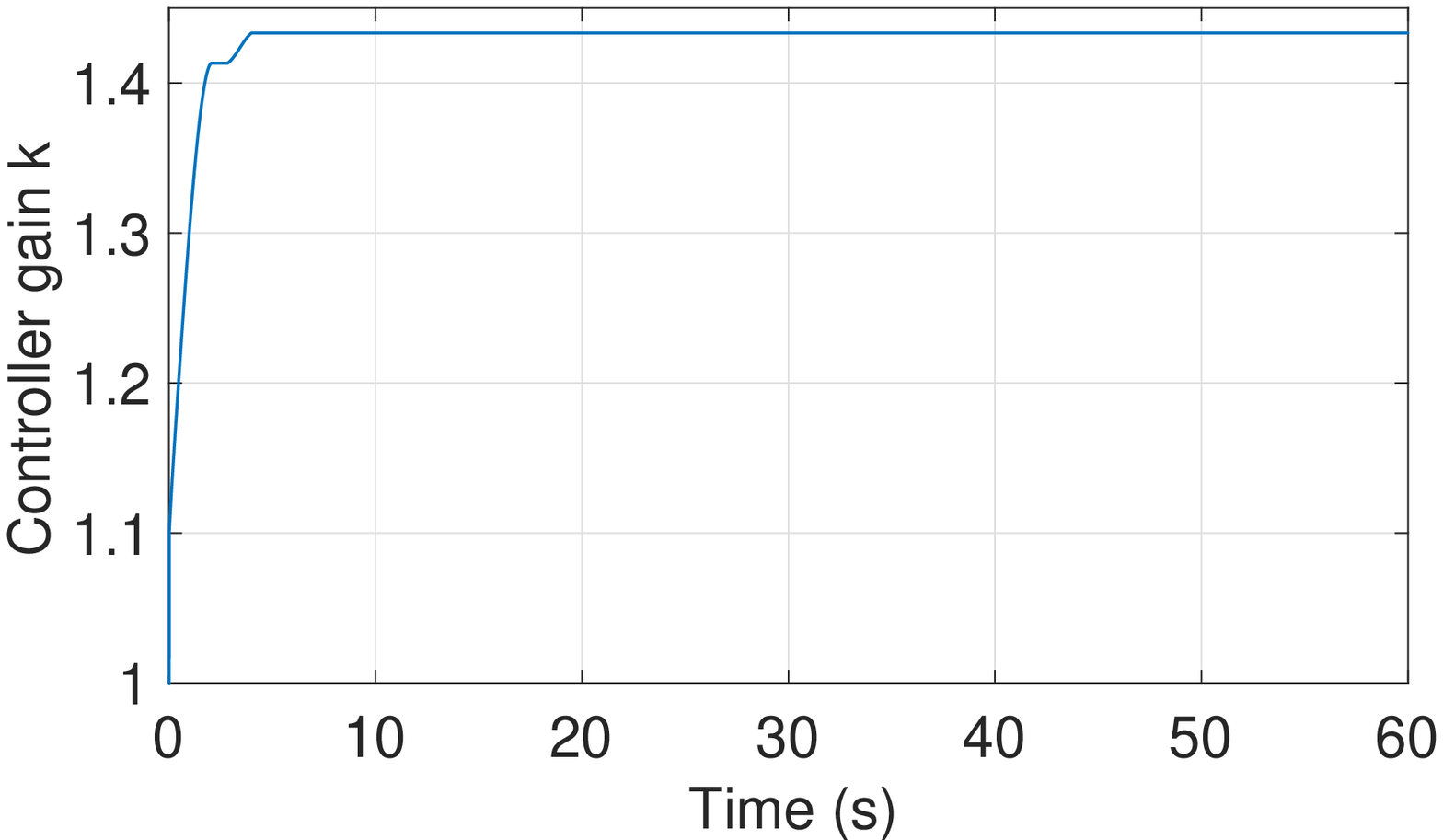}
\label{fig_k}
}
\subfigure[ ]{
\includegraphics[width=0.47\textwidth]{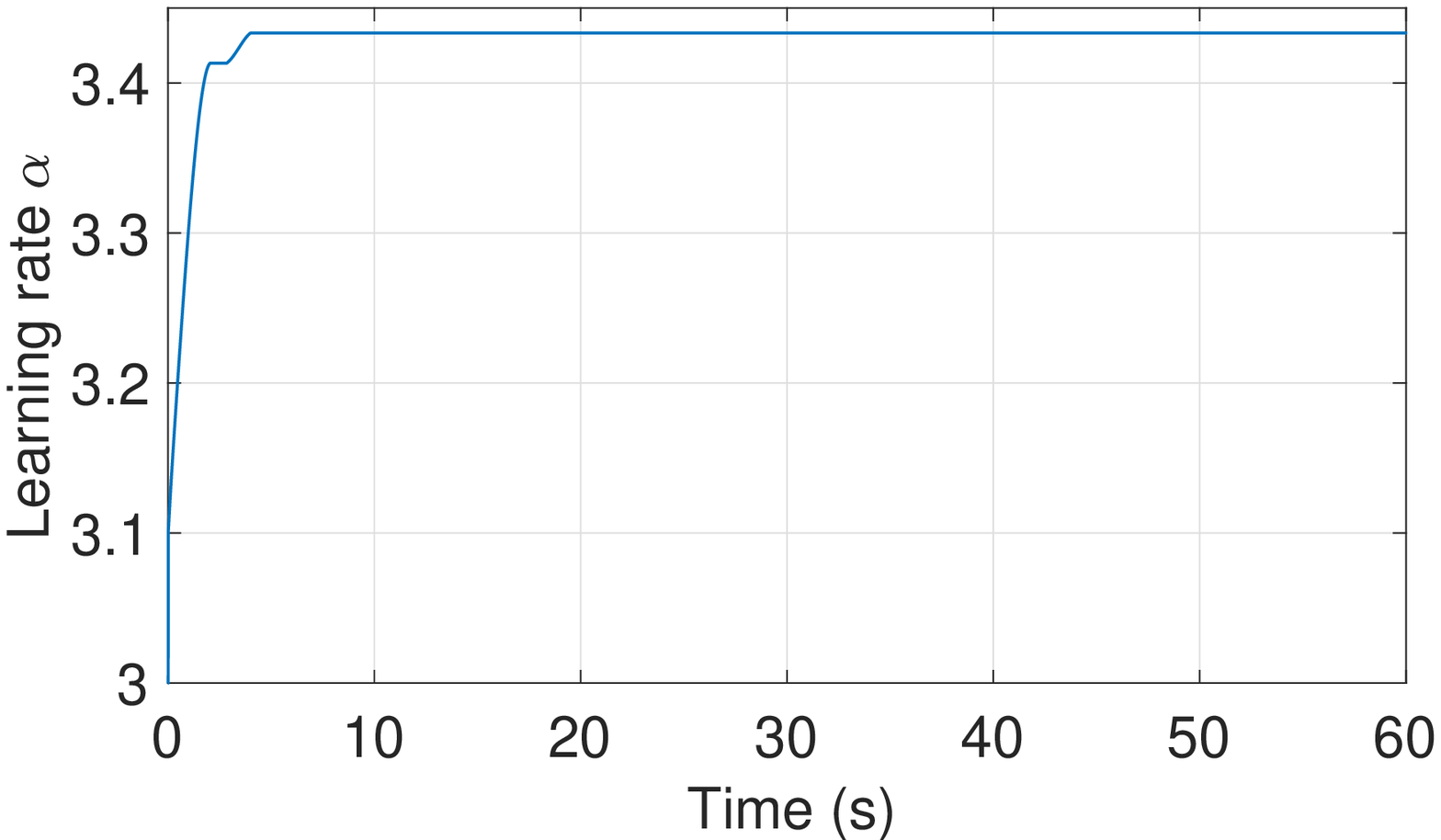}
\label{fig_alpha}
}
\subfigure[ ]{
\includegraphics[width=0.47\textwidth]{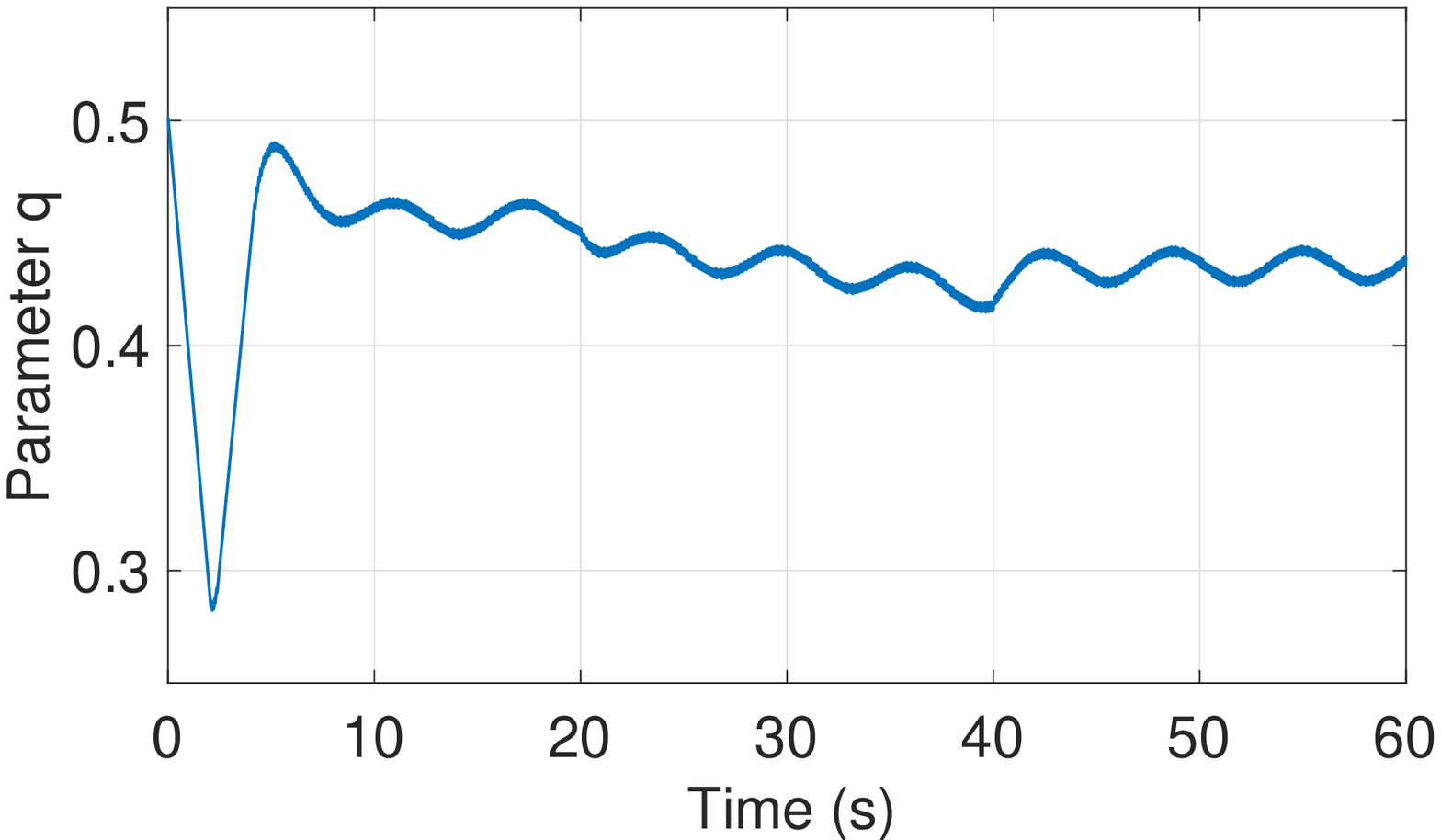}
\label{fig_q}
}
\caption[Optional caption for list of figures]{Control performance on an adaptive cruise control system: (a) Position (b) Speed (c) Acceleration (d) Error (e) Control signals (f) Adaptation of the controller gain (g) Adaptation of the learning rate (g) Adaptation of the parameter q}
\label{sensors}
\end{figure*}

\subsection{Scenerio 2: Noise Analysis}\label{section_num_ex}

To evaluate control performance under noisy condition, the developed SMLC algorithm is applied to the following numerical system \cite{Yang2013}
\begin{eqnarray}\label{}
\dot{x}_{1}& = & x_{2} \nonumber \\
\dot{x}_{2} & = & - 2 x_{1} - x_{2} + e^{x_{1}}+ u
\end{eqnarray}

The initial conditions on the system states are set to $x(0)=[1, -1]^{T}$. The slope of the sliding surface $\lambda$ is set to $2$. The coefficients $\gamma_{k}$ and $\gamma_{\alpha}$ for the adaptions of the controller gain in \eqref{eq_robustterm_k} and the learning rate in \eqref{alpha} are respectively set to $1$ and $0.1$.

In order to test the robustness of this control approach, the states $x_{1}$ and $x_{2}$ are measured with a noise level signal-to-noise ratio (SNR) = $50$ $dB$. The states' responses are shown in Figs. \ref{fig_x1} and \ref{fig_x2}. The T2NFC can learn system behavior online and thus regulates the system without any steady-state error. Thanks to the proposed T2NFC structure, our control algorithm ensures robust control performance against noisy measurements. 

The control signals are shown in Fig. \ref{fig_x1x2controlsignals}. As seen in this figure, the generated control signal by the robust control action $u_c$ converges zero while the generated control signal by the T2NFC $u_n$ is different from zero. It is concluded that the T2NFC can take the overall control of the system in a very short time period. 

The adaptations of the controller gain $k$, the learning rate $\alpha$ and the parameter $q$ are presented in Figs. \ref{fig_x1x2k}-\ref{fig_x1x2q}. The initial conditions of these parameters are considered as $k(0)=1$, $\alpha(0)=0.03$ and $q(0)=0.5$. Thanks to these adaptations, the controller gain, the learning rate and the proportions of the upper and lower membership functions are updated throughout the simulations. Moreover, the controller gain converges to a certain value under noisy conditions owing to the dead-zone.

\begin{figure*}[t!]
\centering
\subfigure[ ]{
\includegraphics[width=0.47\textwidth]{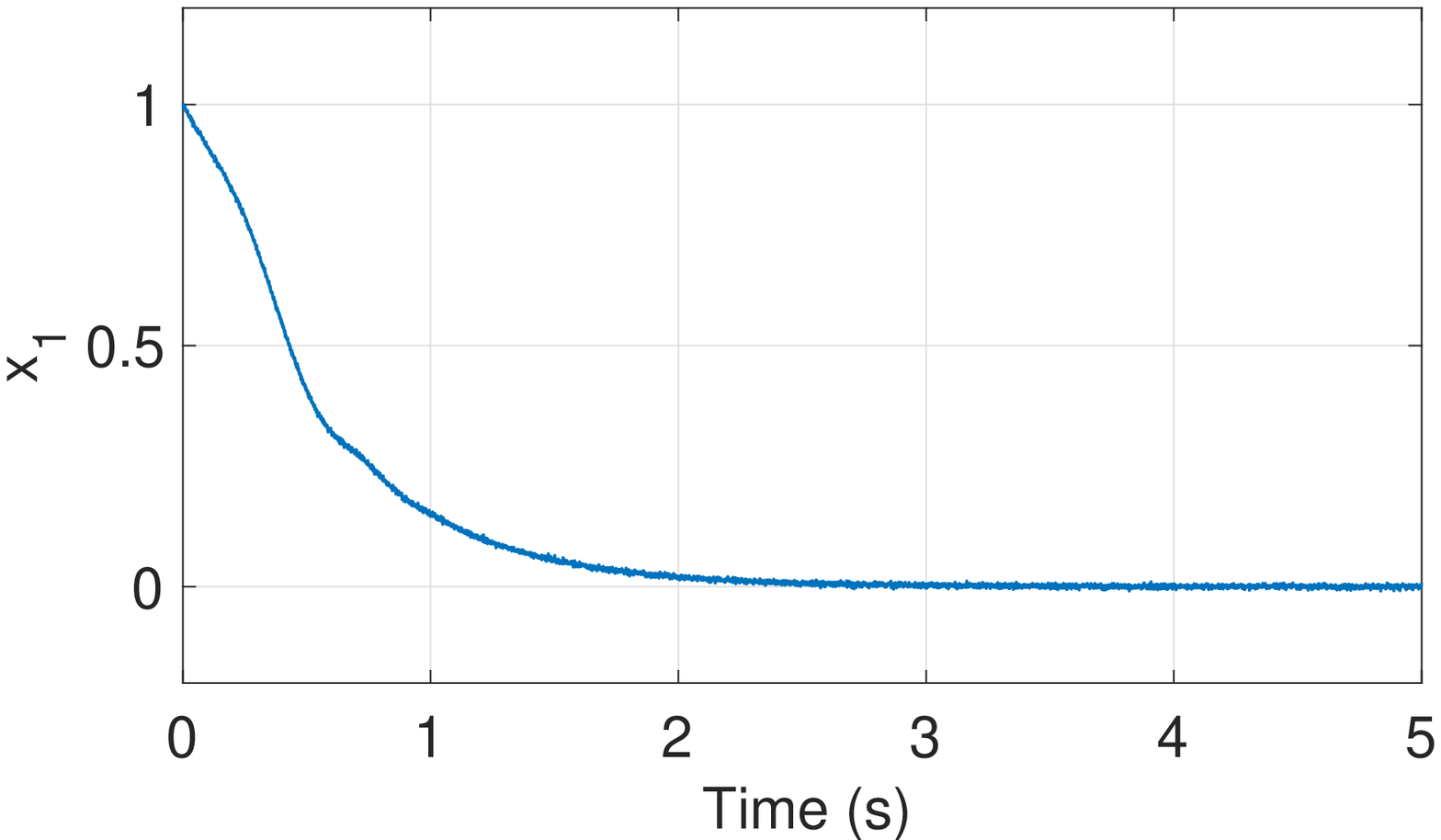}
\label{fig_x1}
}
\subfigure[ ]{
\includegraphics[width=0.47\textwidth]{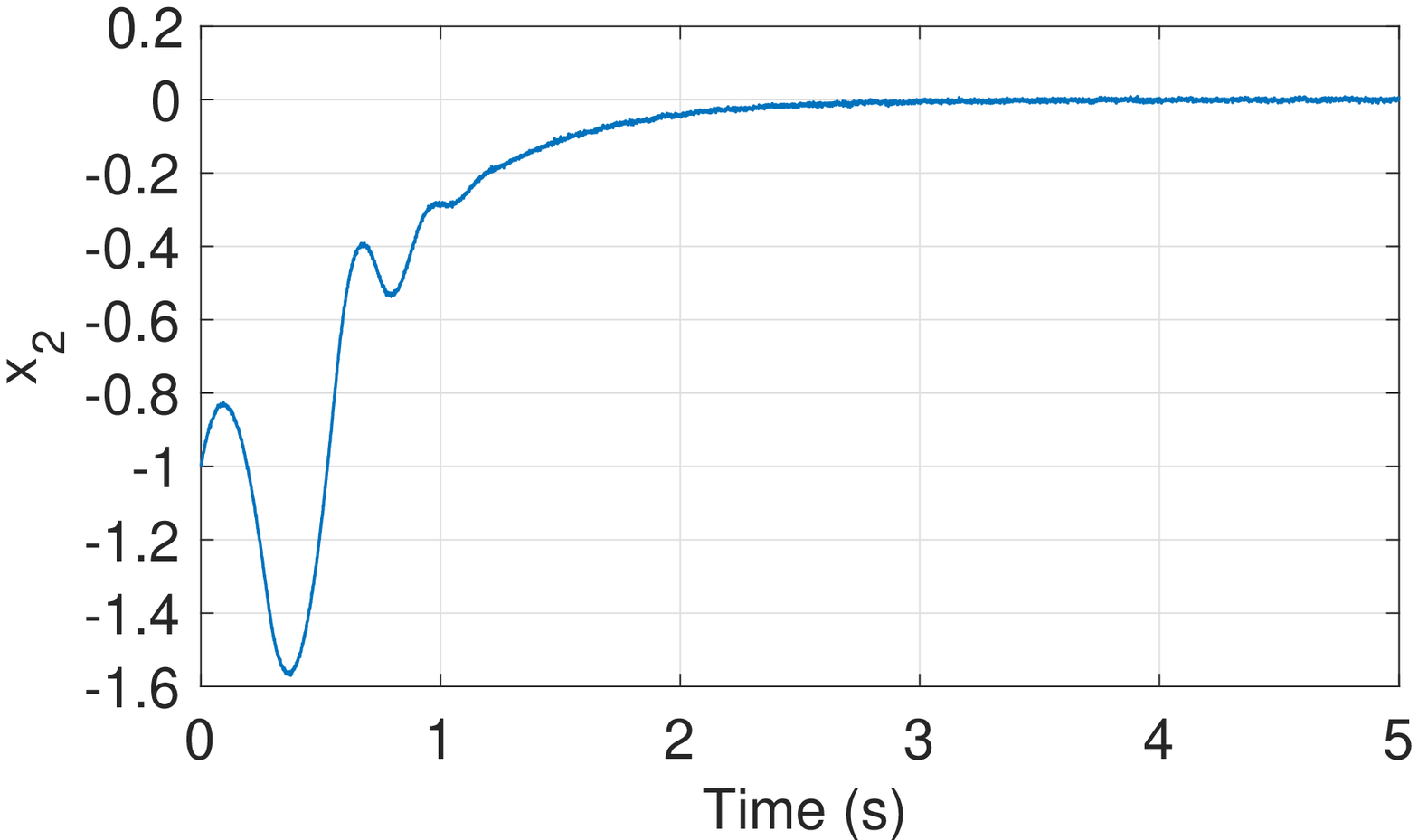}
\label{fig_x2}
}
\subfigure[ ]{
\includegraphics[width=0.47\textwidth]{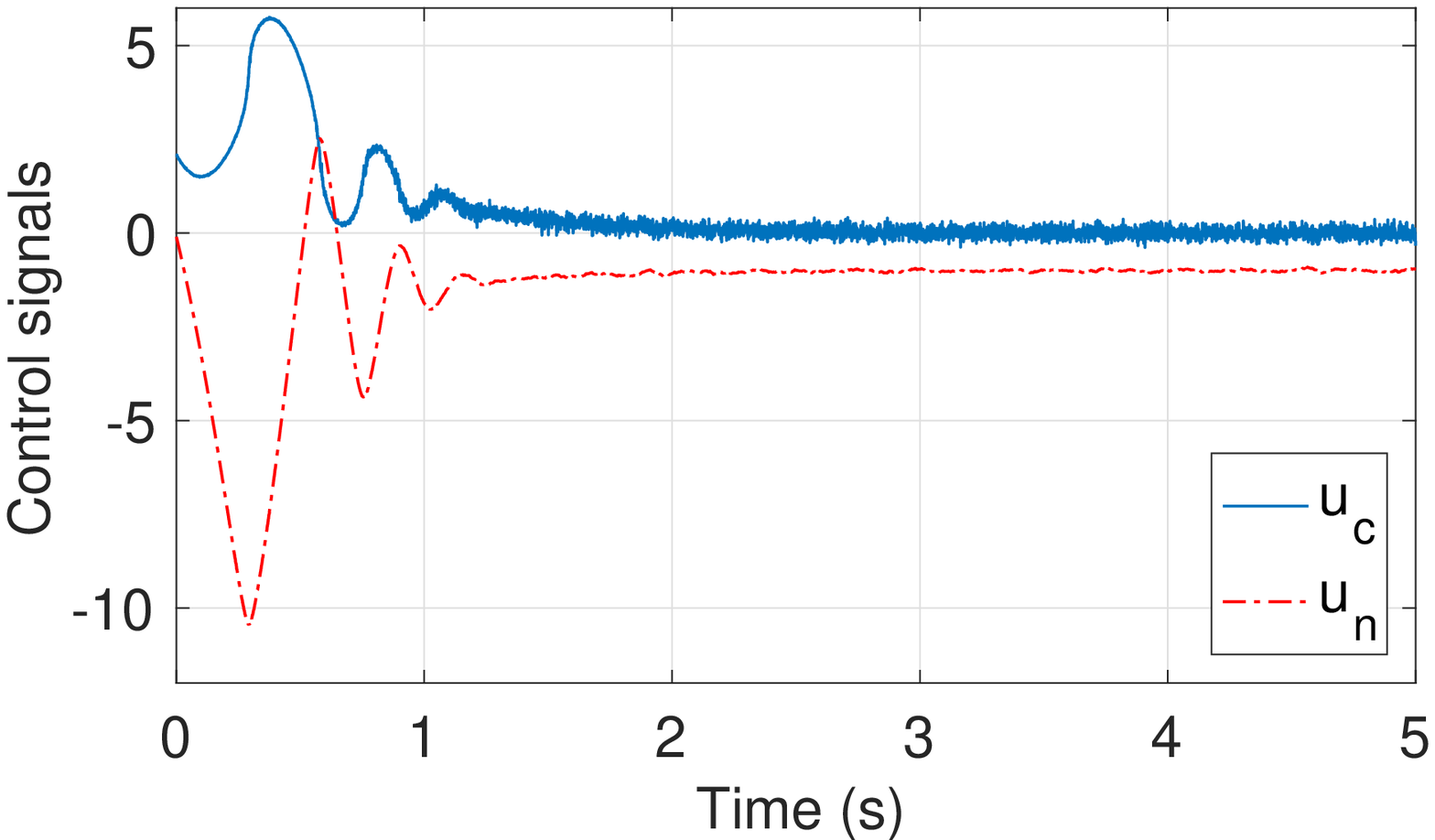}
\label{fig_x1x2controlsignals}
}
\subfigure[ ]{
\includegraphics[width=0.47\textwidth]{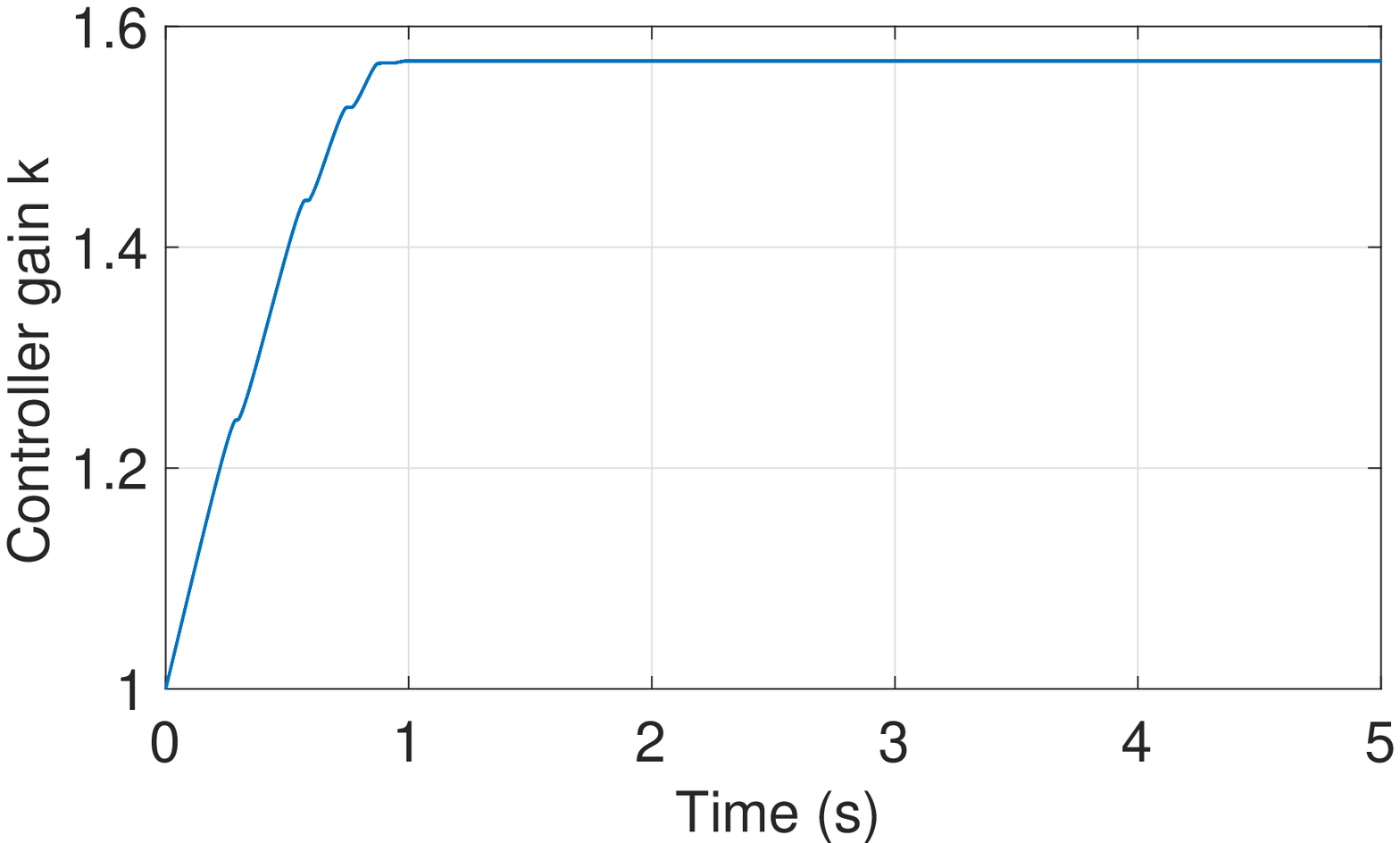}
\label{fig_x1x2k}
}
\subfigure[ ]{
\includegraphics[width=0.47\textwidth]{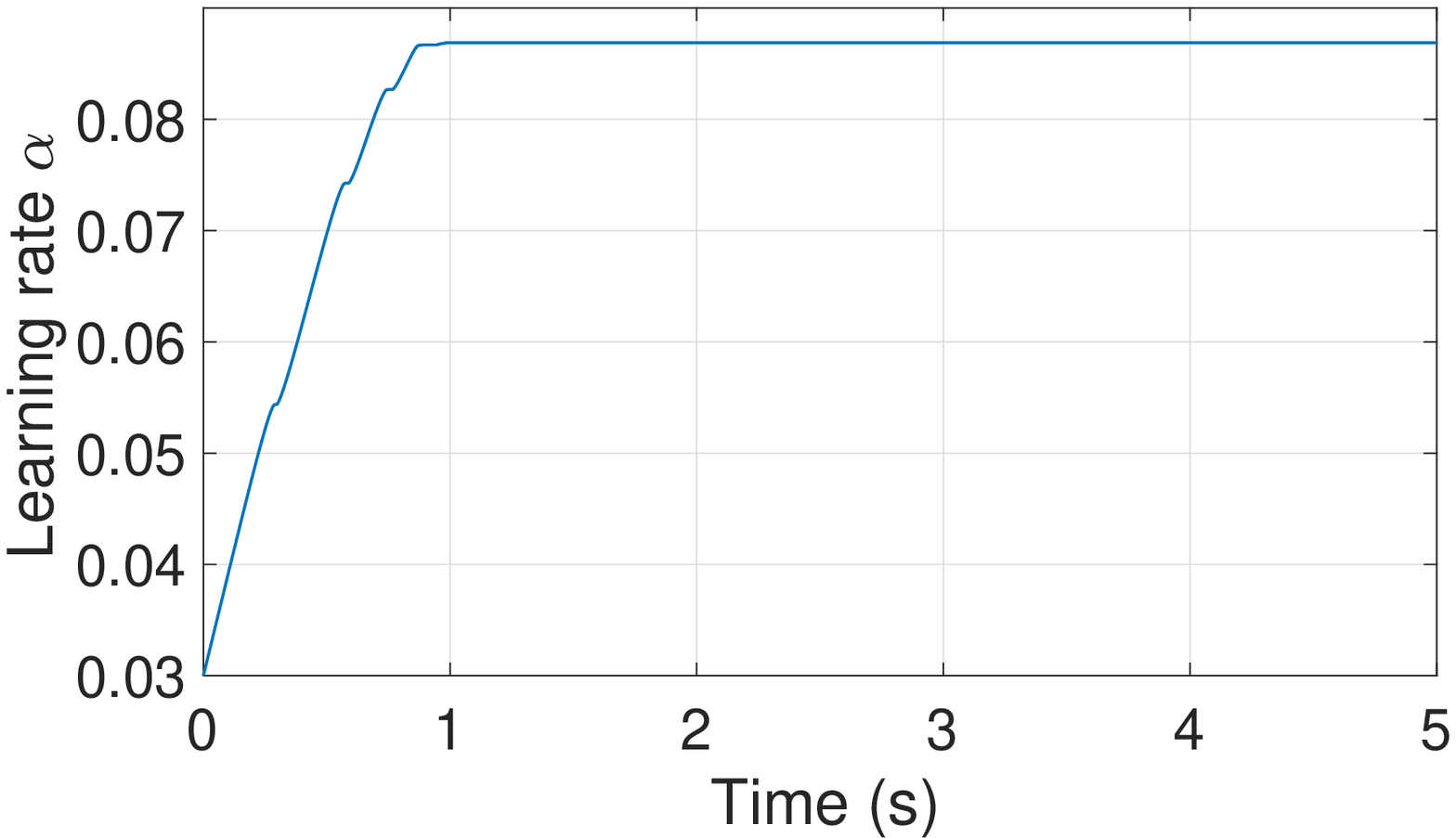}
\label{fig_x1x2alpha}
}
\subfigure[ ]{
\includegraphics[width=0.47\textwidth]{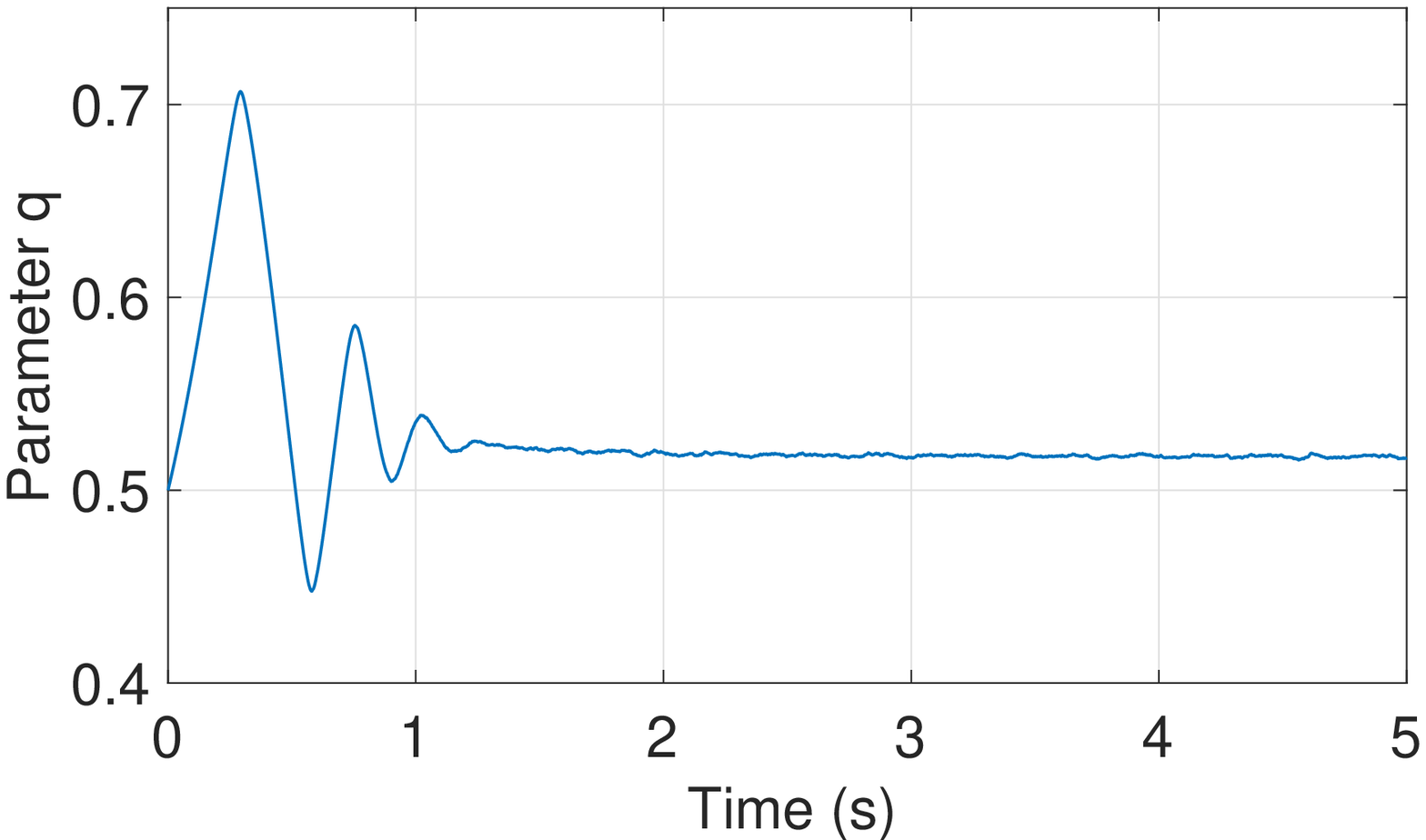}
\label{fig_x1x2q}
}
\caption[Optional caption for list of figures]{Control performance under noisy measurements: (a) State $x_{1}$ (b) State $x_{2}$ (c) Control signals (d) Adaptation of the controller gain (e) Adaptation of the learning rate (f) Adaptation of the parameter q}
\label{sensors}
\end{figure*}

\section{Conclusion}\label{section_conc}

The developed SMLC algorithm for uncertain nonlinear systems has been investigated in this paper. In addition to the stability of the training algorithm, the analytical proof of the overall system stability has been achieved for nth-order uncertain nonlinear systems. The simulation results show that the developed control structure exhibits robust control performance in the presence of external disturbances and noisy measurements. The parameters of the T2NFC controller are automatically regulated through the sliding mode learning algorithm so that T2NFC can learn the system behavior and take the overall control of systems in a very short time period. The proposed sliding mode learning algorithm is valid for only single-input systems. As a future work, the extension of the proposed learning algorithm for multi-input systems is  an interesting topic to be investigated.

\appendix
\section{The time-derivative of the output of the T2NFC $\dot{u}_{n}$}

The following equations are obtained by taking the time derivative of \eqref{mu1_lower}-\eqref{mu2_upper}:
\begin{eqnarray}
\dot{\underline{\mu}}_{1i}(x_1) = -2 \underline{N}_{1i} \dot{\underline{N}}_{1i}\underline{\mu}_{1i}(x_1), \quad 
\dot{\overline{\mu}}_{1i}(x_1) = -2 \overline{N}_{1i} \dot{\overline{N}}_{1i}\overline{\mu}_{1i}(x_1)
\end{eqnarray}
\begin{eqnarray}
\dot{\underline{\mu}}_{2j}(x_2) = -2 \underline{N}_{2j} \dot{\underline{N}}_{2j}\underline{\mu}_{2j}(x_2), \quad
\dot{\overline{\mu}}_{2j}(x_2) = -2 \overline{N}_{2j} \dot{\overline{N}}_{2j}\overline{\mu}_{2j}(x_2)
\end{eqnarray}
where $\underline{N}_{1i}=\Big(\frac{x_1-\underline{c}_{1i}}{\underline{\sigma}_{1i}}\Big)$,
$\overline{N}_{1i}=\Big(\frac{x_1-\overline{c}_{1i}}{\overline{\sigma}_{1i}}\Big)$, $\underline{N}_{2j}=\Big(\frac{x_2-\underline{c}_{2j}}{\underline{\sigma}_{2j}}\Big)$, and
$\overline{N}_{2j}=\Big(\frac{x_2-\overline{c}_{2j}}{\overline{\sigma}_{2j}}\Big)$.

By taking the time derivative of \eqref{wij_lower_upper_normalized}, the following equations are obtained as follows:
\begin{eqnarray}
\dot{\widetilde{\underline{w}}}_{ij}  &=&  \frac{\Big(\underline{\mu}_{1i}(x_1) \underline{\mu}_{2j}(x_2)\Big)'\Big(\sum_{i=1}^{I}\sum_{j=1}^{J}\underline{w}_{ij}\Big)}{(\sum_{i=1}^{I}\sum_{j=1}^{J}\underline{w}_{ij})^2} - \frac{\big(\underline{w}_{ij}\big)\Big(\sum_{i=1}^{I}\sum_{j=1}^{J}\underline{\mu}_{1i}(x_1)   \underline{\mu}_{2j}(x_2)\Big)'}{(\sum_{i=1}^{I}\sum_{j=1}^{J}\underline{w}_{ij})^2} \nonumber \\
\dot{\widetilde{\overline{w}}}_{ij}  &=&  \frac{\Big(\overline{\mu}_{1i}(x_1) \overline{\mu}_{2j}(x_2)\Big)'\Big(\sum_{i=1}^{I}\sum_{j=1}^{J}\overline{w}_{ij}\Big)}{(\sum_{i=1}^{I}\sum_{j=1}^{J}\overline{w}_{ij})^2}  - \frac{\big(\overline{w}_{ij}\big)\Big(\sum_{i=1}^{I}\sum_{j=1}^{J}\overline{\mu}_{1i}(x_1)   \overline{\mu}_{2j}(x_2)\Big)'}{(\sum_{i=1}^{I}\sum_{j=1}^{J}\overline{w}_{ij})^2}
\end{eqnarray}

Since $\widetilde{\underline{w}}_{ij} = \frac{\underline{w}_{ij}}{\sum_{i=1}^{I}\sum_{j=1}^{J}\underline{w}_{ij}}$,
\begin{eqnarray}\label{dotwij_lower_normalized}
\dot{\widetilde{\underline{w}}}_{ij} & = & \frac{\Big(\dot{\underline{\mu}}_{1i}(x_1) \underline{\mu}_{2j}(x_2)+ \underline{\mu}_{1i}(x_1) \dot{\underline{\mu}}_{2j}(x_2)\Big)}{(\sum_{i=1}^{I}\sum_{j=1}^{J}\underline{w}_{ij})}  -\frac{(\widetilde{\underline{w}}_{ij})\bigg(\sum_{i=1}^{I}\sum_{j=1}^{J} \Big(\dot{\underline{\mu}}_{1i}(x_1) \underline{\mu}_{2j}(x_2)+ \underline{\mu}_{1i}(x_1) \dot{\underline{\mu}}_{2j}(x_2)\Big)\bigg)}{(\sum_{i=1}^{I}\sum_{j=1}^{J}\underline{w}_{ij})} \nonumber \\
& = &  \frac{\Big(-2 \underline{N}_{1i} \dot{\underline{N}}_{1i} \underline{\mu}_{1i}(x_1)\underline{\mu}_{2j}(x_2) -2 \underline{N}_{2j} \dot{\underline{N}}_{2j} \underline{\mu}_{1i}(x_1) \underline{\mu}_{2j}(x_2)\Big)}{(\sum_{i=1}^{I}\sum_{j=1}^{J}\underline{w}_{ij})}   -\frac{(\widetilde{\underline{w}}_{ij})\sum_{i=1}^{I}\sum_{j=1}^{J} \Big(-2 \underline{N}_{1i} \dot{\underline{N}}_{1i} \underline{\mu}_{1i}(x_1)\underline{\mu}_{2j}(x_2) \Big)}{(\sum_{i=1}^{I}\sum_{j=1}^{J}\underline{w}_{ij})} \nonumber \\
& &-\frac{(\widetilde{\underline{w}}_{ij})\sum_{i=1}^{I}\sum_{j=1}^{J} \Big(-2 \underline{N}_{2j} \dot{\underline{N}}_{2j} \underline{\mu}_{1i}(x_1) \underline{\mu}_{2j}(x_2)\Big)}{(\sum_{i=1}^{I}\sum_{j=1}^{J}\underline{w}_{ij})} \nonumber \\
&=&  -\widetilde{\underline{w}}_{ij}\dot{\underline{K}}_{ij} + \widetilde{\underline{w}}_{ij} \sum_{i=1}^{I}\sum_{j=1}^{J}\widetilde{\underline{w}}_{ij} \dot{\underline{K}}_{ij}
\end{eqnarray}
where
\begin{equation*}
\dot{\underline{K}}_{ij} = 2 \Big(\underline{N}_{1i} \dot{\underline{N}}_{1i} +\underline{N}_{2j} \dot{\underline{N}_{2j}} \Big)
\end{equation*}

In a similar way, $\dot{\widetilde{\overline{w}}}_{ij} $ obtained:
\begin{equation}\label{dotwij_upper_normalized}
\dot{\widetilde{\overline{w}}}_{ij} = -\widetilde{\overline{w}}_{ij}\dot{\overline{K}}_{ij} + \widetilde{\overline{w}}_{ij} \sum_{i=1}^{I}\sum_{j=1}^{J}\widetilde{\overline{w}}_{ij}  \dot{\overline{K}}_{ij} 
\end{equation}
where
\begin{equation*}
\dot{\overline{K}}_{ij}  = 2 \Big(\overline{N}_{1i} \dot{\overline{N}}_{1i} +\overline{N}_{2j} \dot{\overline{N}}_{2j} \Big)
\end{equation*}

The time derivative of the output of the T2NFC algorithm is obtained as:
\begin{eqnarray}
\dot{u}_n & = & q\sum_{i=1}^{I}\sum_{j=1}^{J}(\dot{f}_{ij} \widetilde{\underline{w}}_{ij}+ f_{ij}\dot{\widetilde{\underline{w}}}_{ij})  +(1-q)\sum_{i=1}^{I}\sum_{j=1}^{J}(\dot{f}_{ij}\widetilde{\overline{w}}_{ij} + f_{ij}\dot{\widetilde{\overline{w}}}_{ij}) +\dot{q} \sum_{i=1}^{I}\sum_{j=1}^{J}f_{ij}\widetilde{\underline{w}}_{ij}
-\dot{q} \sum_{i=1}^{I}\sum_{j=1}^{J}f_{ij}\widetilde{\overline{w}}_{ij}
\end{eqnarray}

If \eqref{dotwij_lower_normalized} and \eqref{dotwij_upper_normalized} are inserted to the aforementioned equation:
\begin{eqnarray}\label{dotVc2}
\dot{u}_n & = & q\sum_{i=1}^{I}\sum_{j=1}^{J}\bigg(\Big(-\widetilde{\underline{w}}_{ij} \dot{\underline{K}}_{ij}+\widetilde{\underline{w}}_{ij} \sum_{i=1}^{I}\sum_{j=1}^{J}\widetilde{\underline{w}}_{ij} \dot{\underline{K}}_{ij} \Big)f_{ij} +\widetilde{\underline{w}}_{ij}\dot{f}_{ij}\bigg) \nonumber \\
&&+(1-q)\sum_{i=1}^{I}\sum_{j=1}^{J}\bigg(\Big( -\widetilde{\overline{w}}_{ij} \dot{\overline{K}}_{ij} + \widetilde{\overline{w}}_{ij} \sum_{i=1}^{I}\sum_{j=1}^{J}\widetilde{\overline{w}}_{ij} \dot{\overline{K}}_{ij} \Big)f_{ij} +\widetilde{\overline{w}}_{ij}\dot{f}_{ij}\bigg) +\dot{q} \sum_{i=1}^{I}\sum_{j=1}^{J}f_{ij} (\widetilde{\underline{w}}_{ij} - \widetilde{\overline{w}}_{ij} )  \nonumber \\
& = &q\sum_{i=1}^{I}\sum_{j=1}^{J}\bigg(\Big(-\widetilde{\underline{w}}_{ij}\big(2\underline{N}_{1i} \dot{\underline{N}}_{1i} +2\underline{N}_{2j} \dot{\underline{N}}_{2j} \big)  + \widetilde{\underline{w}}_{ij} \sum_{i=1}^{I}\sum_{j=1}^{J}\widetilde{\underline{w}}_{ij} \big(2\underline{N}_{1i} \dot{\underline{N}}_{1i} +2\underline{N}_{2j} \dot{\underline{N}}_{2j} \big)\Big) f_{ij}+\widetilde{\underline{w}}_{ij}\dot{f}_{ij}\bigg) \nonumber \\
&& +(1-q)\sum_{i=1}^{I}\sum_{j=1}^{J}\bigg(\Big( -\widetilde{\overline{w}}_{ij}\big(2 \overline{N}_{1i} \dot{\overline{N}}_{1i} +2 \overline{N}_{2j} \dot{\overline{N}}_{2j} \big)  + \widetilde{\overline{w}}_{ij} \sum_{i=1}^{I}\sum_{j=1}^{J}\widetilde{\overline{w}}_{ij} \big(2 \overline{N}_{1i} \dot{\overline{N}}_{1i} +2 \overline{N}_{2j} \dot{\overline{N}}_{2j} \big)\Big)f_{ij} +\widetilde{\overline{w}}_{ij}\dot{f}_{ij}\bigg) \nonumber \\
&& +\dot{q} \sum_{i=1}^{I}\sum_{j=1}^{J}f_{ij} (\widetilde{\underline{w}}_{ij} - \widetilde{\overline{w}}_{ij} )
\end{eqnarray}
where
\begin{eqnarray*}
\dot{\underline{N}}_{1i}= \frac{(\dot{x_1} - \dot{\underline{c}}_{{1i}})\underline{\sigma}_{1i}-(x_1 - \underline{c}_{1i})\dot{\underline{\sigma}}_{1i}}{\underline{\sigma}^2 _{1i}}, \quad
\dot{\underline{N}}_{2i}= \frac{(\dot{x_2} - \dot{\underline{c}}_{{2i}})\underline{\sigma}_{2i}-(x_2 - \underline{c}_{2i})\dot{\underline{\sigma}}_{2i}}{\underline{\sigma}^2 _{2i}}
\end{eqnarray*}
\begin{eqnarray*}
\dot{\overline{N}}_{1i}= \frac{(\dot{x_1} - \dot{\overline{c}}_{{1i}})\overline{\sigma}_{1i}-(x_1 - \overline{c}_{1i})\dot{\overline{\sigma}}_{1i}}{\overline{\sigma}^2 _{1i}}, \quad
\dot{\overline{N}}_{2i}= \frac{(\dot{x_2} - \dot{\overline{c}}_{{2i}})\overline{\sigma}_{2i}-(x_2 - \overline{c}_{2i})\dot{\overline{\sigma}}_{2i}}{\overline{\sigma}^2 _{2i}}
\end{eqnarray*}

\eqref{up8} can be obtained as follows:
\begin{equation}\label{up8}
\underline{N}_{1i}\dot{\underline{N}}_{1i}=\underline{N}_{2i}\dot{\underline{N}}_{2i}=\overline{N}_{1i}\dot{\overline{N}}_{1i}=\overline{N}_{2i}\dot{\overline{N}}_{2i}= \alpha sgn{(s)}
\end{equation}

\eqref{up8} is inserted to \eqref{dotVc2}:
\begin{eqnarray}\label{dotVc3}
\dot{u}_n & = &  q \sum_{i=1}^{I}\sum_{j=1}^{J}\bigg(2\Big(-\widetilde{\underline{w}}_{ij} f_{ij} 2 \alpha sgn{(s)}  + \widetilde{\underline{w}}_{ij} f_{ij} \sum_{i=1}^{I}\sum_{j=1}^{J}\widetilde{\underline{w}}_{ij} 2 \alpha sgn{(s)}\Big)+\widetilde{\underline{w}}_{ij} \dot{f}_{ij} \bigg) \nonumber \\
&& +(1-q)\sum_{i=1}^{I}\sum_{j=1}^{J}\bigg(2\Big(-\widetilde{\overline{w}}_{ij} f_{ij} 2 \alpha sgn{(s)}   +\widetilde{\overline{w}}_{ij} f_{ij} \sum_{i=1}^{I}\sum_{j=1}^{J}\widetilde{\overline{w}}_{ij} 2\alpha sgn{(s)}\Big)\widetilde{\overline{w}}_{ij} \dot{f}_{ij} \bigg)  +\dot{q} \sum_{i=1}^{I}\sum_{j=1}^{J}f_{ij} (\widetilde{\underline{w}}_{ij} - \widetilde{\overline{w}}_{ij} )  \nonumber \\
& = & -4q\alpha sgn{(s)} \sum_{i=1}^{I}\sum_{j=1}^{J}\Big(\widetilde{\underline{w}}_{ij} f_{ij} - \widetilde{\underline{w}}_{ij} f_{ij} \sum_{i=1}^{I}\sum_{j=1}^{J}\widetilde{\underline{w}}_{ij} \Big)  -4 (1-q) \alpha sgn{(s)}\sum_{i=1}^{I}\sum_{j=1}^{J}\Big(\widetilde{\overline{w}}_{ij} f_{ij}  -\widetilde{\overline{w}}_{ij} f_{ij}\sum_{i=1}^{I}\sum_{j=1}^{J}\widetilde{\overline{w}}_{ij} \Big)\nonumber \\
&& +q \sum_{i=1}^{I}\sum_{j=1}^{J}\widetilde{\underline{w}}_{ij} \dot{f}_{ij}  +(1-q)\sum_{i=1}^{I}\sum_{j=1}^{J}\widetilde{\overline{w}}_{ij} \dot{f}_{ij}  +\dot{q} \sum_{i=1}^{I}\sum_{j=1}^{J}f_{ij} (\widetilde{\underline{w}}_{ij} - \widetilde{\overline{w}}_{ij} ) 
\end{eqnarray}

Since $\sum_{i=1}^{I}\sum_{j=1}^{J}\widetilde{\overline{w}}_{ij} =1$ and $\sum_{i=1}^{I}\sum_{j=1}^{J}\widetilde{\underline{w}}_{ij} =1$,  $\dot{u}_{n}$ becomes by using \eqref{f_ij} and \eqref{q} as below:
\begin{eqnarray}\label{dotVc4}
\dot{u}_n & = & q \sum_{i=1}^{I}\sum_{j=1}^{J}\widetilde{\underline{w}}_{ij} \dot{f}_{ij} + (1-q)\sum_{i=1}^{I}\sum_{j=1}^{J}\widetilde{\overline{w}}_{ij} \dot{f}_{ij} +\dot{q} \sum_{i=1}^{I}\sum_{j=1}^{J}f_{ij} (\widetilde{\underline{w}}_{ij} - \widetilde{\overline{w}}_{ij} )  \nonumber \\
& = & 2 \alpha sgn{(s)}
\end{eqnarray}

\bibliographystyle{SageH} 
\bibliography{reference.bib}

\end{document}